\documentclass[11pt]{article}

\usepackage{amsmath, amssymb, amsfonts}
\usepackage{mathtools}
\usepackage{graphicx}
\usepackage{epstopdf}
\usepackage{dsfont}
\usepackage{geometry}
\usepackage{caption}
\usepackage{float}
\usepackage{xcolor}
\usepackage{bm}
\usepackage{enumitem} 
\usepackage{lipsum}
\usepackage[numbers]{natbib}

\mathtoolsset{showonlyrefs}

\DeclareGraphicsExtensions{.eps,.pdf,.png,.jpg}

\usepackage{amsthm}
\usepackage{thmtools} 
\usepackage{hyperref}
\usepackage[capitalise]{cleveref}

\newtheorem{theorem}{Theorem}[section]
\newtheorem{lemma}[theorem]{Lemma}
\newtheorem{proposition}[theorem]{Proposition}
\newtheorem{corollary}[theorem]{Corollary}
\newtheorem{definition}[theorem]{Definition} 
\newtheorem{assumption}[theorem]{Assumption}
\theoremstyle{remark}
\newtheorem{remark}[theorem]{Remark}

\newcommand{\dd}{\mathrm d}
\newcommand{\cA}{\mathcal A}
\newcommand{\cB}{\mathcal B}

\newcommand{\cG}{\mathcal G}

\newcommand{\cL}{\mathcal L}

\newcommand{\cN}{\mathcal N}
\newcommand{\cZ}{\mathcal Z}

\newcommand{\cX}{\mathcal{X}}
\newcommand{\bbR}{\mathbb R} 
\newcommand{\sP}{\mathsf P}

\newcommand{\sE}{\mathsf E}
\newcommand{\cE}{\mathcal E}
\newcommand{\cR}{R}
\newcommand{\wmin}{w_{\mathrm{min}}}
\newcommand{\Var}{\mathrm{Var}}

\newcommand{\qmin}{q_{\mathrm{adj}}}
\newcommand{\rmin}{r_{\mathrm{min}}}
\newcommand{\K}{K}
\newcommand{\T}{T}

\renewcommand{\mid}{\,|\,} 

\newcommand{\ind}{\mathds{1}}  
\newcommand{\Gap}{\lambda}

\newcommand{\TV}{\mathrm{TV}}
\newcommand{\barg}{\bar{g}}
\newcommand{\barP}{\bar{P}}
\newcommand{\KL}{\mathrm{KL}}
\newcommand{\one}{\mathbf{1}_d}

\title{Polynomial Mixing Times of Simulated Tempering for Mixture Targets by Conductance Decomposition}
\author{Quan Zhou\\
Department of Statistics, Texas A\&M University 
}
\date{}
\numberwithin{equation}{section}

\begin{document}

\maketitle

\begin{abstract} 
We study the theoretical complexity of simulated tempering for sampling from mixtures of log-concave components differing only by location shifts. 
The main result establishes the first polynomial-time guarantee for simulated tempering combined with the Metropolis-adjusted Langevin algorithm (MALA) with respect to the problem dimension $d$, maximum mode displacement $D$, and logarithmic accuracy $\log \epsilon^{-1}$.  
The proof builds on a general state decomposition theorem for $s$-conductance, applied to an auxiliary Markov chain constructed on an augmented space. 
We also obtain an improved complexity estimate for simulated tempering combined with random-walk Metropolis. 
Our bounds assume an inverse-temperature ladder with smallest value $\beta_1 = O(D^{-2})$ and spacing $\beta_{i+1}/\beta_i = 1 + O( d^{-1/2} )$, both of which are shown to be asymptotically optimal up to logarithmic factors. 
\end{abstract}

\section{Introduction} \label{sec:intro}
The problem of sampling from a log-concave distribution has been extensively studied in the recent literature, with numerous algorithms now known to achieve  polynomial-time convergence; see, e.g.,~\citet{chewiBook} for a detailed exposition in book form. 
In contrast, sampling from non-log-concave distributions is generally much more challenging. 
A prototypical example is a mixture distribution with log-concave components, where naive algorithms such as the random-walk Metropolis (RWM) are known to mix slowly~\citep{bovier2016metastability}. 
In this work, we study the complexity of the simulated tempering algorithm~\citep{marinari1992simulated} for sampling from a target probability density on $\bbR^d$ given by
\begin{equation} \label{eq:intro-target}
    \pi^*(x) \propto   \sum_{j=1}^{\K}  w_j e^{ -f (x - \mu_j)},   
\end{equation}
where $f$ is $L$-smooth and $m$-strongly convex,  $(w_j)_{j = 1}^{\K}$ are the mixture weights, and 
$(\mu_j)_{j = 1}^{\K} $ are the modes of component distributions. The key idea behind simulated tempering is to augment the target density to a joint density $\pi^*(i, x)$, where for each $i$, $\pi^*(i, \cdot ) \propto [ \pi^*(\cdot) ]^{\beta_i}$ with inverse-temperature parameter $\beta_i > 0$. 
We refer to the sequence $0 < \beta_1 <   \cdots < \beta_{\T} = 1$ as the (inverse) temperature ladder. 
At each temperature level, one can employ any standard sampling algorithm, such as RWM or Metropolis-adjusted Langevin algorithm (MALA), to sample from $\pi^*(i, \cdot)$. 
When $\beta_i$ is small (i.e., the temperature is high), the distribution $\pi^*(i, \cdot)$ is flattened, allowing the sampler to move easily between the modes. By alternating between updating the location $x$ and the temperature index $i$, simulated tempering is able to efficiently explore the entire state space. 

\subsection{Main Results} \label{sec:intro-main-results} 
We prove in \cref{th:main} that no matter whether RWM or MALA  is used at each fixed temperature level, simulated tempering can achieve a mixing time of order $\mathrm{poly}(d, D, \log \epsilon^{-1})$, where $d$ is the problem dimension,   $D = \max_j \| \mu_j \|$ denotes the maximum mode displacement (we always use $\| \cdot \|$ to denote the $L^2$-norm of a vector), and  $\epsilon$ denotes the target  accuracy in  total variation (TV) distance.  
In particular, when RWM is used and all other parameters are fixed,  $O(d^2 \log^2 D \log \epsilon^{-1} )$ iterations are needed to be $\epsilon$-close to the stationary density $\pi^*(i, x)$ in TV distance; see Remark~\ref{rmk:order}. 
This polynomial bound requires the step size of the proposal scheme (of either RWM or MALA) to be properly tuned and that 
\begin{equation} \label{eq:opt-choice-beta}
    \beta_1 \leq \frac{1}{4 L D^2},  \quad \frac{\beta_{i+1}}{\beta_i} = 1 + \frac{m}{L \sqrt{d} }, \text{ for } i = 1, \dots, \T - 1.  
\end{equation}

To our knowledge, our results significantly improve existing complexity bounds for simulated tempering with mixture targets in three key aspects. 
First, the non-asymptotic complexity of simulated tempering combined with MALA has not been studied in the literature, and this work is the first to provide such a theoretical guarantee for the use of MALA in simulated tempering. 
Second, no prior analysis has established polynomial dependence   on all three complexity parameters, $d, D, \log \epsilon^{-1}$. The seminal work of~\citet{woodard2009conditions} only considered a mixture of two Gaussian distributions with means $c \one$ and $-c \one$ for some fixed $c > 0$, where 
$\one =  (1, \dots, 1)$, and thus did not investigate the dependence on $D$. The work of~\citet{ge2018simulated, lee2018beyond} studied  a discretized version of simulated tempering with continuous-time 
Langevin diffusion at each temperature level, yielding a complexity bound that scales polynomially in $\epsilon^{-1}$ instead of $\log \epsilon^{-1}$. 
Third, the conditions on the temperature ladder required by our results also improve those used in existing results. 
For comparison, assuming $L, m$ are fixed,  \citet{ge2018simulated} required $\beta_{i+1} / \beta_i = 1 + O( d^{-1} )$, and~\citet{woodard2009conditions} required $\beta_{i+1} / \beta_i = 1 + O( d^{-1} \log d )$, while our condition~\eqref{eq:opt-choice-beta} only assumes $\beta_{i+1} / \beta_i = 1 + O( d^{-1/2} )$, which corresponds to a smaller number of temperature levels.  
The only work that used the same asymptotic order as~\eqref{eq:opt-choice-beta} is~\citet{garg2025restricted} (see Lemma 19 therein), but their analysis was restricted to mixtures of Gaussian distributions rather than general log-concave components. 
In fact, as we prove in \cref{coro:necessary}, our condition~\eqref{eq:opt-choice-beta} achieves asymptotically optimal dependence on $d$ and $D$, up to logarithmic factors.

\subsection{Proof Overview} \label{sec:intro-proof} 

While our overall proof strategy  follows the general approach of existing works, our analysis features two major technical improvements. First, all existing non-asymptotic mixing time bounds for simulated tempering are obtained via spectral gap analysis. The simulated tempering chain  is decomposed into restricted chains targeting each component distribution and a projected chain characterizing the probability flow across   temperature levels and mixture components. Then, one bounds the spectral gaps of restricted chains and projected chain separately, which yield a multiplicative bound on the overall spectral gap of the  algorithm via the celebrated state decomposition theorem~\citep{woodard2009conditions}.  
However, when MALA is used at each temperature level, this strategy appears difficult to apply, since establishing a lower bound on the spectral gap of MALA on unbounded spaces remains challenging~\citep{bou2013nonasymptotic, chewi2021optimal}. Meanwhile, recent studies on the complexity of MALA for log-concave targets~\citep{dwivedi2019log,  wu2022minimax, chewi2021optimal} are all based on $s$-conductance analysis, where $s$-conductance essentially characterizes the conductance of a Markov chain after excluding sets with small stationary probabilities.  
This motivates us to directly decompose the $s$-conductance rather than the spectral gap of the simulated tempering chain. We show that a minor adaptation of the argument of~\citet{jerrum2004elementary} yields a general state decomposition theorem for $s$-conductance; see \cref{th:conduct-decomp}. 

Another refinement of our analysis lies in the application of state decomposition theorem to general mixture targets. Since~\citet{woodard2009conditions} assumed a mixture of two Gaussian distributions with means $c\one$ and $-c\one$, the state space $\bbR^d$ can be naturally partitioned into two symmetric subsets, and each restricted chain is defined on one subset. 
For general mixture targets of the form given in~\eqref{eq:intro-target}, how to partition the state space becomes a highly non-trivial task~\citep{lee2018beyond}. As an alternative, \citet{ge2018simulated} proposed to directly decompose the Dirichlet form of the simulated tempering chain and compare it with the Dirichlet forms of the samplers at each temperature level and a suitably defined projected chain. 
Their setup assumes that the continuous-time Langevin diffusion is used at each level, leading to a relatively simple expression for the decomposition bound. 
In~\citet{garg2025restricted}, this Markov chain decomposition argument was further extended to discrete-time settings via a more elaborate treatment. In this work, we reconcile the two approaches by comparing the Dirichlet form of the simulated tempering chain to that of an auxiliary Markov chain on the space $\{1,  \dots, \T\} \times \{1, \dots, \K\} \times \bbR^d$ (recall that $\T$ is the number of temperature levels, and $\K$ is the number of mixture components). 
Denote the transition kernel of the actual simulated tempering chain by $P^*$ and that of the auxiliary chain by $P$. The chain $P^*$ has stationary density $\pi^*(i, x)$, while $P$ has stationary density $\pi(i, j, x)$, whose marginal  $\pi(i, x)$ is a mixture distribution that approximates $\pi^*(i, x)$. 
We prove in \cref{lm:compare-st-one} and \cref{lm:compare-st-two} that the $s$-conductance of $P$ can be used to lower bound that of $P^*$. 
The application of the state decomposition theorem to $P$ is straightforward: we simply define each restricted chain on $\{i\} \times \{j\} \times \bbR^d$ for each $(i, j)$. 
This new approach provides a very general and transparent framework for analyzing simulated tempering with mixture target distributions and can be used to study the complexity of simulated tempering combined with other local sampling algorithms such as the proximal sampler~\citep{chen2022improved}. 

One technical challenge is that, in order to obtain a comparison bound for the Dirichlet forms of $P$ and $P^*$, we must assume that at each  level, $P$ and $P^*$ are Metropolis--Hastings chains with the same proposal scheme, though they have different stationary distributions. Hence, for the analysis of simulated tempering combined with MALA, each restricted chain of $P$ is not exactly a MALA sampler: the potential used in the proposal differs from that of the stationary density. This prevents us from directly applying existing $s$-conductance bounds for MALA with log-concave targets~\citep{chewi2021optimal, wu2022minimax}. In \cref{lm:mala}, we bound the $s$-conductance of such ``pseudo-MALA'' chains using an argument similar to that of~\citet{dwivedi2019log}, which is known to be suboptimal for log-concave targets. Whether the use of MALA can yield a faster convergence rate than RWM in simulated tempering remains an open problem.

\subsection{Structure of the Paper}
The rest of the paper is organized as follows. In \cref{sec:general-theory}, we review the definitions and properties of spectral gap and $s$-conductance and then prove a general state decomposition theorem for $s$-conductance. 
In \cref{sec:st}, we formally describe the setup for our theoretical analysis of the simulated tempering algorithm and introduce auxiliary Markov chains used in the proof. 
In \cref{sec:poly}, we state and prove the polynomial complexity bounds for simulated tempering combined with either RWM or MALA.
In \cref{sec:opt-beta}, we construct explicit examples which show that in general, the choice of the temperature ladder given in~\eqref{eq:opt-choice-beta} cannot be improved in terms of the dependence on $d$ and $D$. 
We conclude the paper with a brief discussion in \cref{sec:disc} on possible directions for future research. 
Preliminary lemmas and standard technical arguments are deferred to the appendices. 

\section{General Results for Bounding s-conductance}\label{sec:general-theory}

\subsection{Spectral Gap and s-conductance} 

Throughout this section, we consider a general  state space  $(\cX, \cB(\cX))$, where $\cX$ is assumed to be Polish and $\cB(\cX)$ denotes the Borel $\sigma$-algebra on $\cX$.    
Given a probability measure $\Pi$ on $(\cX, \cB(\cX))$, we write $\Pi(g) = \int_{\cX} g(x) \Pi(\dd x)$  and $V_\Pi(g) = \int_{\cX} [g(x) - \Pi(g)]^2 \Pi(\dd x) $ for any integrable   $g$. 
Denote by $\pi$ the density of $\Pi$ with respect to some dominating measure.  
We recall the definitions of spectral gap and $s$-conductance of a Markov chain. 

\begin{definition}\label{def:gap-conductance}
Let $P$ be the transition kernel of a Markov chain on $(\cX, \cB(\cX))$ reversible with respect to $\Pi$. Denote its spectral gap by $\Gap(P)$, which is defined by 
\begin{equation}
    \Gap(P) = \inf_{g \colon V_\Pi(g) > 0} \frac{ \cE_P(g)}{V_\Pi(g)},  
    \text{ where }
     \cE_P(g)  = \frac{1}{2} \int_{\cX^2} [g(y) - g(x)]^2  \Pi(\dd x)P(x, \dd y). 
\end{equation} 
We call $\cE_P$ the Dirichlet form.  
For $s \in [0, 1/2)$, the $s$-conductance of $P$ is defined by 
\begin{equation}\label{eq:def-s-conduct}
    \Phi_s(P) =  \inf \left\{ \frac{ P(A, A^c) }{ \Pi(A) -s }\colon A \in \cB(\cX) \text{  and }   \Pi(A) \in (s, 1/2] \right\}, 
\end{equation} 
where $P(A, A^c)  = \int_{ A}   P(x, A^c)  \Pi(\dd x)$. 
\end{definition}

It is well known that  $\Gap(P)$ provides a lower bound on the exponential rate at which the chain converges to the stationary distribution~\citep[Chap. 22]{douc2018markov}. An analogous bound can be obtained by using the $s$-conductance. 
To formally state this result,  let $\| \Pi_1 - \Pi_2 \|_{\TV} = \sup_{A \in \cB(\cX)} | \Pi_1(A) - \Pi_2(A)|$ denote the TV distance between two probability measures, and let $\nu P^t$ denote the distribution of the Markov chain (with transition kernel $P$) at time $t$ given initial distribution $\nu$; that is, $(\nu P^t)(A) = \int_{\cX} P^t(x, A) \nu(\dd x)$ for each $A \in \cB(\cX)$. It was shown in~\citet{lovasz1993random}  
that if $P$ is lazy (i.e., $P(x, \{x\}) \geq 1/2$ for all $x$) and reversible with respect to $\Pi$, then for any  distribution $\nu$ and $t \geq 0$, 
\begin{equation}\label{eq:s-conduct-tv-rate}
    \| \nu P^t - \Pi \|_{\TV} \leq \eta s + \eta e^{ -t \Phi_s^2(P)/ 2 }, 
    \text{ where } \eta = \sup_{A \in \cB(\cX)} \frac{ \nu(A)}{\Pi(A)}. 
\end{equation}
Hence, if $\eta < \infty$ (in which case we say $\nu$ is an $\eta$-warm start and call $\eta$ the warm-start parameter), for any target accuracy level $\epsilon > 0$, one can choose a sufficiently small $s$ to show that the TV distance drops below $\epsilon$ at an exponential rate as $t \rightarrow \infty$. 

Some useful inequalities for comparing $s$-conductance and spectral gap are given in \cref{lm:s-conduct}.  
A common approach to bounding the $s$-conductance of a Markov chain is to use the isoperimetric inequality~\citep{lovasz1993random, dwivedi2019log}, provided that such an inequality holds for the stationary distribution. We summarize this method in \cref{lm:s-conductance-isoperi}.

\subsection{Comparing s-conductance of Metropolis--Hastings Chains}\label{sec:compare-mh}
Let $\pi$ be the density of a probability measure on $(\cX, \cB(\cX))$ and $Q$ be a transition kernel with density $q(x, y)$.  
We say $P$ is the transition kernel of a Metropolis--Hastings chain with stationary density $\pi$ and proposal $Q$ if 
\begin{equation}\label{eq:def-MH-general}
    P(x, \dd y) = Q(x, \dd y) a_P(x, y)  \, + \, \left( 1 - \int_{\cX}    a_P(x, z) Q(x, \dd z) \right) \delta_x (\dd y ),  
\end{equation}
where $\delta_x$ denotes the Dirac measure at $x$, and $a_P(x, y) $ is the acceptance probability   given by 
\begin{equation}\label{eq:def-acc}
    a_P(x, y) = \min\left\{ 1, \frac{\pi( y) q(y, x)}{\pi(x) q(x,  y)}\right\}.  
\end{equation}
 
Comparison results for the spectral gaps of two Metropolis--Hastings chains sharing the same proposal are well known~\citep{madras1999importance}. Here we prove an analogous result for $s$-conductance. 

\begin{lemma}  \label{lm:compare}
For $k = 1, 2$, let $P_k$ be the transition kernel of the Metropolis--Hastings chain on $(\cX, \cB(\cX))$ with stationary density $\pi_k$ and same proposal $Q$. 
Suppose for some $c \in (0, 1]$, 
\begin{equation}
    c \leq \frac{\pi_1(x)}{\pi_2(x)} \leq c^{-1}, \quad \forall \, x \in \cX. 
\end{equation}
Then,   
    $\Phi_s(P_1) \geq    c^2  \Phi_{c s} (P_2)$ for any $s \in [0, 1/2)$ and $\Gap(P_1) \geq c^2 \Gap(P_2)$.    
\end{lemma}
\begin{proof}
The inequality for spectral gaps directly follows from 
Proposition 2.3 of~\citet{madras1999importance}. 
To prove the inequality for $s$-conductance, fix a set $A$ with $\Pi_1(A) \in (s, 1/2]$.  
The assumption on $\pi_1 /\pi_2$ implies that $\Pi_1(A) \leq c^{-1}\Pi_2(A)$. 
Hence, if $\Pi_2(A) \leq 1/2$, we get  
\begin{align*}
    \Pi_1(A) - s \leq   \frac{1}{c} (\Pi_2(A) - c s) 
    \leq    \frac{P_2(A, A^c)}{c \, \Phi_{c s} (P_2) } \leq \frac{ P_1(A, A^c)}{c^2 \Phi_{c s} (P_2) }, 
\end{align*}
where the last inequality follows from the assumption on $\pi_1/\pi_2$ and that 
\begin{equation}
    P_k(A, A^c) = \int_{x \in A} \int_{y \in A^c}\min \left\{ q(x, y) \pi_k(x),  \, q(y, x) \pi_k(y) \right\} \dd x \, \dd y.  
\end{equation}

If $\Pi_2(A) \geq 1/2$, then we bound $\Pi_1(A)$ by  $\Pi_1(A) \leq \Pi_1(A^c) \leq c^{-1} \Pi_2(A^c)$. Applying the same argument and using $P_k(A, A^c) = P_k(A^c, A)$, again we get $\Pi_1(A) - s \leq P_1(A, A^c) /  ( c^2 \Phi_{c s} (P_2) )$. 
The claimed bound then follows from the definition of $s$-conductance. 
\end{proof}

\subsection{State Decomposition Theorem for s-conductance}
For spectral gap, various state decomposition bounds have been established~\citep{madras2002markov, jerrum2004elementary, guan2007small, pillai2017elementary, qin2025spectral}. 
The main idea is to decompose the state space into disjoint subsets and then analyze how the chain moves within each piece and between the pieces.   
For a positive integer $n$, we use the shorthand $[n] = \{1, 2, \dots, n\}$. 

\begin{definition}\label{def:decomp} 
Let $P$ be a transition kernel on $(\cX, \cB(\cX))$ that is reversible with respect to the distribution $\Pi$.  
Given a disjoint partition $(\cX_k)_{k \in [n]}$ of $\cX$, define the  transition kernel  $P_k$  of the ``restricted chain'' on $\cX_k$ by 
\begin{equation}
    P_k(x, A)  =   P(x, A \cap \cX_k) 
    + \delta_x(A) P(x, \cX \setminus \cX_k), \quad \forall x \in \cX_k, \, A \in \cB(\cX). 
\end{equation}
Define the transition kernel $\barP$ of the ``projected chain'' on $[n]$ by 
\begin{equation}
    \barP(k, \ell) = \frac{ \int_{\cX_k} P(x, \cX_\ell) \Pi (\dd x) }{ \Pi(\cX_k) }. 
\end{equation}
\end{definition}
 
By checking the detailed balance condition, one can verify that $P_k$ and $\bar{P}$ are reversible with respect to the distribution  $\Pi_k$ and  density $\bar{\pi}$ respectively, where 
\begin{equation}\label{eq:def-muk-barmu}
    \Pi_k(A) = \frac{ \Pi(A \cap \cX_k) }{ \Pi(\cX_k)}, \quad \bar{\pi}(k) = \Pi(\cX_k), 
\end{equation}
for any  $A \in \cB(\cX)$ and $k \in [n]$. Note that $\bar{\pi}$ is a probability mass function on $[n]$, and $\Pi_k$ has support $\cX_k$.  
Existing state decomposition theorems  bound  $\Gap(P)$ using  $\Gap(\barP)$ and $\min_k \Gap(P_k)$, capturing the intuition that fast mixing within each $\cX_k$ and across components leads to fast global mixing. In particular, Theorem 5.2 of~\citet{woodard2009conditions} gives
\begin{equation}\label{eq:gap-decomp}
    \Gap(P)  \geq \frac{1}{2} \, \Gap(\barP) \min_{k \in [n]} \Gap(P_k). 
\end{equation}  

We show in the following theorem that a similar bound can be obtained for the $s$-conductance of $P$.    The assumption $\Gap(\barP) \leq 1$ is mild, since it holds for any lazy $P$ (in which case $\barP$ is also lazy).

\begin{theorem}\label{th:conduct-decomp}
Consider the setting of~\cref{def:decomp} and assume $\Gap(\barP) \leq 1$, and $\Pi(A_0) = 1/2$ for some $A_0 \in \cB(\cX)$. 
For any $s \in [0, 1/2)$,   
\begin{equation}\label{eq:decomp-temp-1}
    \Phi_s(P) \geq \frac{\Gap(\barP)}{8}  \Psi\left( \frac{\Gap(\barP)}{8}  s \right),  
\end{equation}
where  
$\Psi(s) = \min_{k} \Phi_s(P_k)$ denotes the minimum $s$-conductance of restricted chains. 
\end{theorem}

\begin{proof} 
Let  $g = \ind_A$ for some arbitrary $A \in \cB(\cX)$ with $\Pi(A) \in (s, 1/2]$. 
It is straightforward to verify that $\cE_P(g) = P(A, A^c)$, and thus it suffices to show that 
\begin{equation}\label{eq:decomp-goal}
    \frac{\cE_P(g) }{\Pi(A) - s} \geq  \frac{\Gap(\barP)}{8}  \Psi\left( \frac{\Gap(\barP)}{8}  s \right). 
\end{equation} 
By \cref{th:jerrum}, 
\begin{align}
     \cE_P(g)  
    &=   \frac{1}{2} \sum_{k \neq \ell} \cE_{P}^{k, \ell}(g) + \sum_{k \in [n]} \bar{\pi}(k) \cE_{P_k}(g), \label{eq:decomp-E-main-text} \\  
    V_ \Pi(g)    &\leq \frac{3}{2 \Gap(\barP)}\sum_{k \neq \ell} \cE_{P}^{k, \ell}(g)   +  \frac{4}{\Gap(\barP)}   \sum_{k \in [n]} \bar{\pi}(k)  V_{ \Pi_k}(g),   \label{eq:decomp-V-main-text}
\end{align}
where $\cE_{P}^{k, \ell}$ is defined in~\eqref{eq:def-Pkl} and  $\Gap(\barP) \leq 1$ is used in the second inequality. 
Letting $\tilde{s} = s \, \Gap(\barP) / 8$ and using $V_ \Pi(g) =  \Pi(A)  \Pi(A^c)$, we get 
\begin{align*}\label{eq:decomp-temp1}
\Pi(A) - s &\leq  
 2  V_ \Pi(g) -  s   \\
    &\leq \frac{3}{ \Gap(\barP)}\sum_{k \neq \ell} \cE_{P}^{k, \ell}(g)  + \frac{8}{\Gap(\barP)} \sum_{k \in [n]} \bar{\pi}(k) V_{ \Pi_k}(g)  - s \\
    &= \frac{3}{ \Gap(\barP)}\sum_{k \neq \ell} \cE_{P}^{k, \ell}(g)  + \frac{8}{\Gap(\barP)} \sum_{k \in [n]} \bar{\pi}(k) \left[ V_{ \Pi_k}(g)  - \tilde{s} \right] \\
    & \leq  \frac{3}{ \Gap(\barP)}\sum_{k \neq \ell} \cE_{P}^{k, \ell}(g)  +  \frac{8}{\Gap(\barP) \,  \Psi( \tilde{s}) } \sum_{k \in [n]} \bar{\pi}(k)  \cE_{P_k}(g), 
\end{align*} 
where the last step follows from   part~(i) of Lemma~\ref{lm:s-conduct}.  
Comparing the right-hand side with~\eqref{eq:decomp-E-main-text}, we find that 
\begin{equation}\label{eq:decomp-min}
    \frac{\cE_P( \ind_A )}{  \Pi(A)  - \tilde{s}} \geq 
    \min\left\{ \frac{\Gap(\barP)}{6}, \,  \frac{ \Gap(\barP ) \Psi(\tilde{s})}{ 8 } \right\}. 
\end{equation} 

Since $\Gap(\barP) \leq 1$ and $s \leq 1/2$, we have $\tilde{s}  = s \, \Gap(\barP) / 8 \leq 1/16$. By parts (ii) and (iii) of Lemma~\ref{lm:s-conduct}, 
\begin{equation}
   \frac{1}{8} \Psi(\tilde{s})  \leq \frac{1}{8} \Psi\left( \frac{1}{16} \right) \leq 
\frac{1}{8} \frac{1}{1 - 2 / 16} = \frac{1}{7}. 
\end{equation}
Hence,  the minimum in~\eqref{eq:decomp-min} is attained by the second term, which yields~\eqref{eq:decomp-goal} and thus proves the claim. 
\end{proof}

A more general version of \cref{th:conduct-decomp} is given in \cref{th:conduct-decomp-general}, which, in certain scenarios, may yield a sharper lower bound on the $s$-conductance of order $ \min\{ \Gap(\barP), \Psi(c s \Gap(\barP)  ) \} $ for some constant $c > 0$.

\section{Theoretical Setup for Simulated Tempering Analysis} \label{sec:st} 

\subsection{Target Distribution} \label{sec:st-target}
Consider a probability density function $\pi^*$ with respect to the Lebesgue measure on $\bbR^d$  given by 
\begin{equation}\label{eq:target-pi-star}
    \pi^*(x) \propto e^{-U(x)}, \quad 
    U(x) = - \log  \sum_{j=1}^{\K}  w_j e^{ -f (x - \mu_j)},   
\end{equation}
where $\K$ is the number of mixture components, $(w_j)_{j=1}^{\K}$ are strictly positive weights that sum to one, $\mu_j \in \bbR^d$ for $j \in [\K]$,  and $f \colon \bbR^d \rightarrow [0, \infty)$. We assume that $f$ satisfies the following condition, which is widely used in the sampling literature. 

\begin{assumption}\label{ass:f}
The function $f$ in~\eqref{eq:target-pi-star} is twice differentiable, $L$-smooth and $m$-strongly convex for  $m, L > 0$, and $f(x)$ is minimized at $x = 0$ with $f(0) = 0$.  Define $\kappa = L / m$. 
\end{assumption}

For the definition of $L$-smoothness and $m$-strong convexity, see \cref{def:smooth-convex}, which essentially mean that $m I_d  \preceq   \nabla^2 f (x) \preceq L I_d$ for each $x$. The quantity $\kappa = L / m$  is often called the condition number.    
Note that assuming $\inf_{x \in \bbR^d} f(x) = f(0) = 0$ introduces no loss of generality, since for any strongly convex function $\tilde{f}$, there always exists a unique minimizer $x^*$, and we can define $f(x) = \tilde{f}(x + x^*) - \tilde{f}(x^*)$. 
Some properties of $L$-smooth and $m$-strongly convex functions and strongly log-concave distributions are given in \cref{sec:prelim-convex}.

Let $\beta_1 \leq \beta_2 \leq \cdots \leq \beta_{\T} = 1$ be a sequence of $\T$ inverse temperatures. 
Simulated tempering is a Metropolis--Hastings chain with state space $[\T] \times \bbR^d$ and stationary density 
\begin{equation}\label{eq:def-pi-star-ix}
    \pi^*(i, x) =    \frac{  r_i  \pi^*(x)^{\beta_i}  }{\int_{\bbR^d} \pi^*(y)^{\beta_i}  \dd y }, 
\end{equation}
where $(r_i)_{i=1}^{\T}$ are strictly positive and sum to one. 
In practice, one directly specifies the ratio $r_i / \textstyle\int  \pi^*(y)^{\beta_i}  \dd y$ instead of $r_i$, since the normalizing constant is typically unknown; see Remark~\ref{rmk:constant} for more discussion. 
Let $\pi_i^*(x)$ denote the conditional distribution of $x$ given $i$, which can be expressed by 
\begin{equation}\label{eq:def-pi-star-i}
    \pi_i^*(x) \propto e^{ - \beta_i U(x)}, 
\end{equation}
where $U$ is the potential of $\pi^*$ defined in~\eqref{eq:target-pi-star}. 

\subsection{Proposal Scheme} \label{sec:st-proposal}
The transition kernel of the proposal scheme of simulated tempering, which we denote by $Q^*$, can be described by  
    \begin{equation}\label{eq:def-Q-st}
        Q^*( (i, x), \, \{i'\} \times A ) =  \alpha \, q(i, i') \delta_x(A) + 
         (1 - \alpha) \, Q_i(x, A)  \ind_{\{i = i'\}}, \quad \forall i' \in[\T], A \in \cB(\bbR^d), 
    \end{equation}
where $\alpha \in (0, 1)$ is a constant, $q(i, i')$ denotes a transition matrix on $[\T]$, and   $Q_i$ is a transition kernel on $(\bbR^d, \cB(\bbR^d))$. 
In words, with probability $\alpha$, we propose changing the temperature level according to $q(i, i')$, and with probability $1 - \alpha$, we propose a new value for $x$ using $Q_i(x, \cdot)$.  
Note that  we never propose changing $i$ and $x$ simultaneously.

We consider two choices of $Q_i$ in this paper, both of which are very common in practice: 
\begin{align} 
    \text{RWM proposal:  } & Q_i(x, \cdot)  = N(x, (2h \beta_i^{-1}) I_d); \\ 
    \text{MALA proposal:  } & Q_i(x, \cdot) = N(x - h   \nabla U(x), (2h \beta_i^{-1}) I_d). 
\end{align} 
For both proposal schemes, the step size is $h / \beta_i$, where the scaling $\beta_i^{-1}$ is introduced to ensure that the acceptance probabilities at different temperature levels are comparable.  
For the MALA proposal, we use  that at the $i$-th temperature level, the stationary density $\pi^*_i$ has potential $\beta_i U(x)$, and with a step size of $h/\beta_i$, this yields the drift term $h \nabla U (x)$ in the proposal density.

\subsection{Transition Kernel} \label{sec:st-mtk}
Let $P^*$ denote the transition kernel of simulated tempering. 
According to the construction of  Metropolis--Hastings chains  in \cref{sec:compare-mh},  $P^*$ is uniquely determined given the stationary density $\pi^*(i, x)$ and the proposal scheme $Q^*$. 
Explicitly, the acceptance probabilities for the two types of proposals can be computed by 
\begin{align}
 a_{P^*}( (i, x), (i, x') ) & = \min\left\{ 1, \, \frac{ \pi^*_i(x') q_i(x', x)}{\pi^*_i( x) q_i(x, x')} \right\},  \label{eq:acc-Pstar-x} \\ 
  a_{P^*}( (i, x), (i', x)  ) &= \min\left\{ 1, \, \frac{ \pi^*(i', x) q(i', i) }{\pi^*(i, x) q(i, i') } \right\}  
  \label{eq:acc-Pstar-i}
\end{align}
where   $q_i(x, x')$ denotes  the density of the kernel $Q_i$. 
Since  either $x$ or $i$ can be updated in one step, we can express $P^*$ as 
\begin{align} 
   P^*( (i, x), \, \{i\} \times A ) &= (1 - \alpha)  \int_{A}  a_{P^*}( (i, x), (i, x') ) Q_i(x, \dd x') +  b(x) \delta_x(A)    \\
    P^*( (i, x), \, \{i'\} \times A ) &= \alpha\,  q(i, i') a_{P^*}( (i, x), (i', x) ) \delta_x(A),  \text{ if }  i \neq i',    
\end{align}
where $b(x)   = 1 -  (1 - \alpha)  \int_{\bbR^d}  a_{P^*}( (i, x), (i, y) ) Q_i(x, \dd y)  - \alpha \sum_{i' \neq i}  q(i, i') a_{P^*}( (i, x), (i', x) ).$ 

\subsection{First Auxiliary Metropolis--Hastings Chain}\label{sec:first-aux-chain}
The simulated tempering chain $P^*$ is not convenient to analyze, since whenever $\beta_i \neq 1$,  $\pi^*_i$ is no longer a mixture distribution. 
This motivates us to construct  auxiliary Metropolis--Hastings chains whose $s$-conductance can be compared to that of $P^*$.   
Define a joint density $\pi(i, j, x)$ on $[\T] \times [\K] \times \bbR^d$ by 
\begin{equation}\label{eq:def-joint-pi}
    \pi(i, j, x) = 
    \frac{ r_i}{C_i}   w_j e^{ -\beta_i \,  f (x - \mu_j)}, \quad \text{ where } C_i = \int_x e^{ -\beta_i \,  f (x)} \dd x. 
\end{equation}
The marginal distribution of $x$ is a mixture  of $\T \times \K$ component distributions: 
\begin{equation}\label{eq:def-pi-ij}
    \pi(x)  = \sum_{i, j} \pi(i, j)  \pi_{i, j}(x),  \text{ where } \pi(i, j) = r_i w_j, \text{ and } \pi_{i, j}(x) = \frac{1}{C_i} e^{ -\beta_i \,  f (x - \mu_j)}. 
\end{equation}
In the above expression, $\pi(i, j)$ is the marginal probability of $(i, j)$, and $\pi_{i, j}(x)$ is the conditional density of $x$ given $(i, j)$. 
Note that the normalizing constant $C_i$ does not depend on the location parameter $\mu_j$. Integrating over $j$, we get the marginal density
\begin{equation}\label{eq:def-approx-pi}
    \pi(i, x) =  \frac{ r_i}{C_i}  \sum_{j \in [\K]}  w_j e^{ -\beta_i \,  f (x - \mu_j)}, 
\end{equation}
and we use $\pi_i(x) = \sum_{j \in [\K]} w_j \pi_{i, j}(x)$ to denote the conditional density of $x$ given $i$. 

The first auxiliary transition kernel, denoted by $\tilde{P}$, is that of the Metropolis--Hastings chain with state space $[\T] \times \bbR^d$,  stationary density $\pi(i, x)$, and proposal $Q^*$. 
The acceptance probabilities of $\tilde{P}$ are thus given by 
\begin{equation} 
 a_{\tilde{P}}( (i, x), (i, x') )    = \min\left\{ 1, \frac{ \pi_i(x') q_i(x', x)}{\pi_i(x) q_i(x, x')} \right\},   \quad  
  a_{\tilde{P}}( (i, x), (i', x)  ) = \min\left\{ 1, \frac{ \pi(i', x) q(i', i)}{\pi(i, x)q(i, i') } \right\}. 
\end{equation}  
This chain has also  been used  in~\citet{ge2018simulated, garg2025restricted} for spectral gap comparison.

\begin{lemma}\label{lm:compare-st-one}  
Define $\wmin = \min_{j \in [\K]} w_j$, and let $P^*$ and $\tilde{P}$ be the transition kernels introduced above. 
The following results hold:  
\begin{enumerate}[label=(\roman*)] 
 \setlength\itemsep{2pt}  
    \item For any $s \in [0, 1/2)$,  $\Phi_s(P^*) \geq  \wmin^2  \Phi_{s \, \wmin  }(\tilde{P}).  $
    \item For any $s \in [0, \wmin / 2)$,  $ \Phi_s(P^*)  \leq \wmin^{-2} \Phi_{s / \wmin} (\tilde{P}) $. 
    \item  $\wmin^{-2} \, \Gap(\tilde{P}) \geq \Gap(P^*) \geq \wmin^2 \, \Gap(\tilde{P}).$ 
\end{enumerate} 
\end{lemma}

\begin{proof}
By Lemma 7.3 of~\citet{ge2018simulated} (see also \cref{lm:ge-compare}), $\pi^*(i, x) / \pi(i, x) \in [\wmin, \, \wmin^{-1}]$. All the three claims then follow from  \cref{lm:compare}.  
\end{proof}

\subsection{Second Auxiliary Metropolis--Hastings Chain}\label{sec:second-aux-chain}
While the chain $\tilde{P}$  targets a mixture density at every temperature level, the mixture component label $j$ is still treated as latent and not explicitly updated in the chain's dynamics. 
So we now introduce another auxiliary Markov chain, which will be used to establish the polynomial mixing times of simulated tempering. 

Let $P$ be the transition kernel of a Metropolis--Hastings chain on the joint space $[\T] \times [\K] \times \bbR^d$ with stationary density $\pi(i, j, x)$ defined in~\eqref{eq:def-joint-pi} and proposal $Q$ defined by 
\begin{equation}\label{eq:def-Q-joint}
   Q( (i, j, x),  \, \{i'\} \times \{j'\} \times A) 
    =  \frac{1}{2}   \pi_{i, x}(j') \ind_{ \{i = i', x \in A \} }   
    + \frac{1}{2}  Q^*( (i, x), \, \{i'\} \times A )  \ind_{\{ j = j'\} }, 
\end{equation}
where $ \pi_{i, x}(j) = \pi(i, j, x) / \pi(i, x)$ denotes the conditional density of $j$ given $(i, x)$.  
In words, with probability $1/2$, we fix $j$ and propose $(i', x')$ according to the transition kernel $Q^*$ used by the actual simulated tempering algorithm;
with probability $1/2$, we fix $(i, x)$ and propose $j'$ by sampling it from the full conditional distribution.  
Updates of $j$ are always accepted as in Gibbs sampling, and the acceptance probabilities of the other two types of moves are given by 
\begin{align}
    a_{P}( (i, j, x), \, (i, j, x') )  &= \min\left\{ 1, \, \frac{ \pi_{i,j}(x') q_i(x', x)}{\pi_{i,j}(x) q_i(x, x')} \right\},  \\
    a_{P}( (i, j, x), \, (i', j, x) )  &= \min\left\{ 1, \, \frac{ \pi(i', j, x) q(i', i)}{\pi(i, j, x) q(i, i')} \right\}. 
\end{align}

\begin{table}[!t] 
\small
\caption{Construction of the three Metropolis--Hastings chains.}\label{tab:mh-compare}
\begin{center}
  \begin{tabular}{cccc} \hline
     & \bf State Space & \bf Stationary Density & \bf Proposal \\ \hline
    $P^*$ & $[\T] \times \bbR^d$ & $\pi^*(i, x)$ & $Q^*$ (either change $i$ or change $x$) \\
    $\tilde{P}$ & $[\T] \times \bbR^d$ & $\pi(i, x)$ & $Q^*$ (either change $i$ or change $x$) \\
    $P$ & $[\T] \times [\K] \times \bbR^d$ & $\pi(i, j, x)$ & $Q$ (change $i$, change $j$, or change $x$) \\  \hline
  \end{tabular}
\end{center}
\end{table}

We summarize the construction of the three chains $P^*, \tilde{P}, P$ in Table~\ref{tab:mh-compare}. Unlike the analysis of $\tilde{P}$, 
we cannot  use \cref{lm:compare} to compare the spectral gap or $s$-conductance of $P$ with that of $P^*$ or $\tilde{P}$, since the  chains are defined on different spaces. 
However, since the stationary distribution of $P$ admits that of $\tilde{P}$ as a marginal, we can compare the Dirichlet forms of $P$ and $\tilde{P}$ for any function $g(i, x)$ that is independent of $j$, which yields the following key lemma.

\begin{lemma}\label{lm:compare-st-two} 
Let $\tilde{P}$  and $P$ be the transition kernels introduced above.  
For any $s\in [0, 1/2)$, $\Phi_s(\tilde{P}) \geq 2 \Phi_s (P)$, and $\Gap(\tilde{P}) \geq 2 \Gap(P)$. 
\end{lemma}

\begin{proof} 
Fix an integrable measurable function $\tilde{g} \colon [\T] \times \bbR^d \rightarrow \bbR$,  and define $g$ on $ [\T] \times [\K] \times \bbR^d $ by $g(i, j, x) = \tilde{g}(i, x)$. 
Consider the Dirichlet form $\cE_P(g)$. Since $P$ involves three types of proposals (change $i$, change $j$, or change $x$) and $g$ is independent of $j$, we can decompose $\cE_P(g)$ as 
\begin{align*}
  \cE_P(g) = \cE_{P, 0} (g) + \cE_{P, 1}(g), 
\end{align*}
where 
\begin{align}
  \cE_{P, 0} (g) &= \frac{1}{2} \sum_{i, i' \in [\T], j \in [\K]}     \int   [ \tilde{g}(i, x) - \tilde{g}(i', x)  ]^2  \pi(i, j, x) \,  \frac{\alpha}{2} q(i, i') a_P( (i, j, x), \, (i', j, x) ) \, \dd x, \\  
  \cE_{P, 1} (g) &= \frac{1}{2} \sum_{i  \in [\T], j \in [\K]}   \int \int    [ \tilde{g}(i, x) - \tilde{g}(i', x)  ]^2  \pi(i, j, x) \,  \frac{1 - \alpha}{2} q_i(x, x') a_P( (i, j, x), \, (i, j, x') ) \, \dd x \, \dd x'. 
\end{align}

By~\cref{lm:compare-acc}, 
\begin{align}
    2\cE_{P, 0}(g) &\leq \frac{1}{2} \sum_{i, i' \in [\T]}    \int   [ \tilde{g}(i, x) - \tilde{g}(i', x)  ]^2  \pi(i, x) \,  \alpha \, q(i, i') a_{\tilde{P}}( (i, x), \, (i', x) )  \, \dd x, \\ 
    2\cE_{P, 1} (g) &\leq \frac{1}{2} \sum_{i  \in [\T]}  \int \int    [ \tilde{g}(i, x) - \tilde{g}(i', x)  ]^2  \pi(i,  x) \,  (1 - \alpha) q_i(x, x') a_{\tilde{P}}( (i,   x), \, (i,   x') ) \, \dd x \, \dd x'. 
\end{align}
Observe that the sum of the right-hand sides  of the above two inequalities equals $\cE_{\tilde{P}}(\tilde{g})$. Hence,  
\begin{equation}\label{eq:compare-two-diri}
\cE_P(g) = \cE_{P, 0} (g) + \cE_{P, 1}(g) \leq \frac{1}{2} \cE_{\tilde{P}}(\tilde{g}). 
\end{equation}

Let $\Pi$ denote the probability measure on   $[\T] \times [\K] \times \bbR^d $ with density $\pi(i, j, x)$, and let $\tilde{\Pi}$ denote the probability measure on $[\T] \times \bbR^d$ with density $\pi(i, x)$.  Since $\pi(i, x)$ is the marginal of $\pi(i, j, x)$, we have $\Pi(g) = \tilde{\Pi}(\tilde{g})$ and $\Var_\Pi(g) = \Var_{\tilde{\Pi}}(\tilde{g})$. 
The claimed inequality for spectral gaps then follows from~\eqref{eq:compare-two-diri} and the definition of spectral gap. By restricting $\tilde{g}$ to indicator functions, we can repeat the same argument to prove the claimed inequality for $s$-conductance. 
\end{proof}

\section{Polynomial Mixing Time Bounds for Simulated Tempering}\label{sec:poly}

\subsection{Main Results and Proof Outline} \label{sec:outline} 

We now state our main results. For ease of presentation, we assume that the density $q(i, i')$ satisfies 
\begin{equation}\label{ass:Q}
    q(i, i + 1) = q(i + 1, i) = \qmin, \; \forall \, i\in [\T - 1], 
\end{equation}
where $\qmin \in (0, 1)$ is a constant. This assumption is standard and usually satisfied in  practice. For the temperature ladder, we use  
\begin{equation}\label{eq:choice-beta}
\T = \lceil (\kappa \sqrt{d} + 1) \log (4 L D^2 + 1) \rceil, \text{  and  } 
\frac{\beta_{i+1}}{\beta_i} =  1 + \frac{1}{\kappa \sqrt{d} } \text{  for } i \in [\T - 1], 
\end{equation}
where we recall $\kappa = L/m$. These conditions  guarantee that $\beta_1 \leq (4 L D^2)^{-1}$.  
We characterize the complexity of the simulated tempering algorithm with respect to the dimension $d$, target accuracy $\epsilon$,  proposal parameters $\alpha$ and $\qmin$,  smoothness parameter $L$, strong convexity parameter $m$, and mixture distribution parameters 
\begin{equation}\label{eq:def-complex-parameters}
D = \max_{j \in [\K]} \| \mu_j \|,   \quad  \wmin = \min_{j \in [\K]} w_j,  \quad
\rmin  = \min_{i \in [\T]} r_i,  \quad \tilde{r} = \min_{i, i' \in [\T]}  \frac{r_{i'} }{ r_i}. 
\end{equation}
Note that if one chooses $\K = 2$, $\mu_1 = \one$ and $\mu_2 = -\one$, then $D = \sqrt{d}$, which grows with the dimension $d$. 
Let $P^*_{\mathrm{lazy}}$ denote the lazy version of $P^*$, i.e., $P^*_{\mathrm{lazy}}(x, A) = (P^*(x, A) + \delta_x(A) ) / 2$ for any $A \in \cB(\bbR^d)$.  

\begin{theorem}\label{th:main}
Consider the simulated tempering chain $P^*$ defined in \cref{sec:st-mtk} with stationary density $\pi^*$ given by~\eqref{eq:target-pi-star}. Suppose Assumption~\ref{ass:f} holds, $q(i, i')$ satisfies~\eqref{ass:Q}, 
$\T$ and $(\beta_i)_{i \in [\T]}$ are given by~\eqref{eq:choice-beta}. 
Define
\begin{equation}
    \tau = \alpha  (1 - \alpha) \, \qmin \, \tilde{r}.  
\end{equation}
Let $\epsilon, \eta > 0$ and $Q_i$ be either of the following two kernels:
\begin{align}
    Q_i(x, \cdot)  &= N(x, (2h \beta_i^{-1}) I_d), \text{ with } h = \frac{1}{L d}, \\ 
    Q_i(x, \cdot) &= N(x - h \nabla U(x), (2h \beta_i^{-1}) I_d), \text{ with } h = \frac{c}{L^2 (D + \cR_{\eta, \epsilon})^2 d}, 
\end{align}
where $c \leq 0.01$ and 
\begin{equation}\label{eq:def-Rs}
    \cR_{\eta, \epsilon} = \frac{2}{\sqrt{m}} \max \left\{ \sqrt{d}, \, \sqrt{ 2 \log \frac{86 \, \T \eta}{ \tau \, \rmin \, \wmin \, \epsilon } } \right\}.  
\end{equation}
Given  an $\eta$-warm initial distribution $\nu$, we have $\| \nu (P^*_{\mathrm{lazy}})^t - \Pi \|_{\TV} \leq \epsilon$ for  
    \begin{equation}
        t \geq  \frac{c' \T^2}{h    m \,  \tau^2 \, \rmin^2 \, \wmin^4 } \log \frac{2 \eta}{\epsilon}, 
    \end{equation}
where $c' > 0$ denotes some universal constant. 
\end{theorem}

\begin{proof}
By the definition of lazy Markov chains, we have $\Phi_{s}( P^*_{\mathrm{lazy}} ) = \Phi_{s}( P^*) / 2$. Combining \cref{lm:compare-st-one} and \cref{lm:compare-st-two}, we get $\Phi_{s}( P^*) \geq \wmin^2 \Phi_{s \, \wmin} (\tilde{P}) 
\geq 2 \wmin^2 \Phi_{s \, \wmin} (P) $ for any $s \in [0, 1/2)$. Hence, by~\eqref{eq:s-conduct-tv-rate}, the claims hold if 
\begin{equation}\label{eq:tv-bound-t}
     \frac{ 2 \log (2 \eta / \epsilon)}{ \Phi^2_{\epsilon / 2\eta}( P^*_{\mathrm{lazy}} )} = \frac{ 8 \log (2 \eta / \epsilon)}{ \Phi^2_{\epsilon / 2\eta}( P^*  )}  \leq  \frac{ 2\log (2 \eta / \epsilon)}{  \wmin^4 \Phi_{s   }^2(P) }   \leq t, \text{ where } s = \frac{\epsilon \, \wmin}{ 2 \eta}. 
\end{equation}

Let $\Theta = [\T] \times [\K] \times \bbR^d$ denote the state space of $P$, which admits the straightforward partition $\Theta = \cup_{i \in [\T], j \in [\K]} \Theta_{i, j}$ with $\Theta_{i, j} = \{ i \} \times \{ j \} \times \bbR^d$. 
Let $\barP$ denote the resulting projected chain and $P_{i, j}$ the restricted chain on $\Theta_{i, j}$. 
By \cref{th:conduct-decomp},
\begin{equation}\label{eq:decomp-st}
    \Phi_s(P) \geq \frac{\Gap(\barP)}{8} \min_{i \in [\T], j \in [\K]} \Phi_{s'} (P_{i, j}), 
\end{equation}
where $s' = s  \Gap(\barP) / 8$. 
We will analyze the spectral gap of $\barP$ in \cref{sec:gap-projected}. The main findings are summarized in \cref{lm:project-chain} and \cref{lm:H-delta}, from which we get  
\begin{equation}\label{eq:gap-barP-mix}
    \Gap(\barP) \geq \frac{3 \, \alpha \, \tilde{r} \, \rmin \, \qmin  }{16 \, \T}.   
\end{equation} 
Therefore, to use~\eqref{eq:decomp-st}, it remains to bound $\Phi_{\tilde{s}}(P_{i, j})$ where 
\begin{equation}
    \tilde{s} = \frac{3\, \alpha \, \tilde{r} \, \rmin \, \qmin  }{16 \, \T}  \frac{s}{8}  =  \frac{3 \, \alpha \, \tilde{r} \, \rmin \, \qmin \, \epsilon \, \wmin  }{256 \, \T \eta}   \geq \frac{\tau \, \rmin \,  \wmin  \epsilon}{86 \, \T \eta}. 
\end{equation}  

Let $c'$ denote a positive universal constant whose value may change between occurrences. 
The analysis of the conductance of $P_{i, j}$ will be detailed in \cref{sec:conduct-restrict}. 
By \cref{lm:rwm} and \cref{lm:mala} therein, the assumption on the step size $h$ ensures that $\Phi_{\tilde{s}}(P_{i, j}) \geq c' (1 - \alpha) \sqrt{h m}$. Hence, \cref{th:conduct-decomp} yields that 
\begin{align}
    \Phi_s(P) \geq \frac{  \alpha \, \tilde{r} \, \rmin \, \qmin  }{40 \T}   \min_{i, j} \Phi_{\tilde{s}} (P_{i, j}) \geq  
    \frac{ c'  \tau \, \rmin }{ \T }  \sqrt{h m}. 
\end{align}
Plugging this bound back into~\eqref{eq:tv-bound-t} completes the proof. 
\end{proof}

\begin{remark}\label{rmk:constant} 
Ideally, one wants to set $r_i = 1 / \T$  for each $i \in [\T]$ so that $\rmin = 1 / \T$ and $\tilde{r} = 1$. But since in practice, the normalizing constant 
\begin{equation}
    Z^*_i = \int_{\bbR^d} \pi^*(y)^{\beta_i} \dd y. 
\end{equation} 
is typically unknown, one implements the simulated tempering algorithm with $Z^*_i$ replaced by some estimate $\hat{Z}^*_i$, which effectively sets $r_i \propto Z^*_i / \hat{Z}^*_i$ by~\eqref{eq:def-pi-star-ix}.  
As shown in~\citet{lee2018beyond}, these estimates $(\hat{Z}^*_i)_{i \in [\T]}$ can be obtained iteratively. First, one runs the simulated tempering algorithm with only one inverse temperature $\beta_1$ and uses the samples to estimate $Z^*_2/ Z^*_1$. Then, one runs the algorithm with inverse temperatures $\beta_1, \beta_2$ to estimate $Z^*_3 / Z^*_2$, and so on. The polynomial-time guarantees provided in \cref{th:main} ensure that only a polynomial number of samples is required to obtain $(\hat{Z}^*_i)_{i \in [\T]}$ such that  $r_i = \Theta(\T^{-1})$ for each $i$ with high probability. We refer readers to~\citet{ge2018simulated, garg2025restricted} for details. 
\end{remark}

For RWM proposal, we can get a better bound by directly decomposing the spectral gap instead of $s$-conductance. Recall that given two probability measures $\Pi_1$ and $\Pi_2$, the $\chi^2$-divergence is defined as $\chi^2( \Pi_2 \,\|\, \Pi_1) = \Var_{\Pi_1} ( \dd \Pi_2 / \dd \Pi_1 )$. 

\begin{corollary}\label{coro:main-rwm}
Consider the setting of \cref{th:main}, and let $Q_i(x, \cdot)  = N(x, (2h \beta_i^{-1}) I_d)$ with $h = 1/(Ld)$. 
Given an initial distribution $\nu$ with $\chi^2( \nu   \, \| \, \Pi ) \leq \tilde{\eta}^2$ for some $\tilde{\eta} > 0$, we have $\| \nu (P^*_{\mathrm{lazy}})^t - \Pi \|_{\TV} \leq \epsilon$ for  
\begin{equation}
        t \geq  \frac{c' d  \kappa  \T}{ \tau  \rmin \,  \wmin^2 } \log \frac{ \tilde{\eta}}{2\epsilon}, 
\end{equation}
where $c' > 0$ denotes some universal constant.  
\end{corollary}
\begin{proof}
Since the TV distance can be bounded by the $\chi^2$-divergence via the Cauchy--Schwarz inequality and spectral gap governs the convergence rate in the $\chi^2$-divergence~\citep{montenegro2006mathematical}, we get 
\begin{align*}
    \| \nu (P^*_{\mathrm{lazy}})^t - \Pi \|_{\TV} \leq \frac{1}{2}  \sqrt{ \chi^2( \nu (P^*_{\mathrm{lazy}})^t \, \|  \, \Pi ) } \leq \frac{1}{2} (1 - \Gap(P^*_{\mathrm{lazy}}))^t \sqrt{ \chi^2( \nu   \, \| \, \Pi ) }. 
\end{align*}
Hence, it suffices to choose  
\begin{equation}\label{eq:coro-rwm-1}
     \frac{1}{\Gap(P^*_{\mathrm{lazy}})} \log \frac{\tilde{\eta}}{ 2 \epsilon} =    \frac{2}{\Gap(P^*)} \log \frac{\tilde{\eta}}{ 2 \epsilon}  \leq  \frac{1}{\wmin^2 \Gap(P)} \log \frac{\tilde{\eta}}{ 2 \epsilon}  \leq t, 
\end{equation}
where we have used \cref{lm:compare-st-one} and \cref{lm:compare-st-two} to get $\Gap(P^*) \geq 2 \wmin^2 \Gap(P)$. 
By~\eqref{eq:gap-decomp}, 
\begin{equation}
    \Gap(P)  \geq \frac{1}{2} \Gap(\barP) \min_{i, j} \Gap(P_{i, j}) \geq \frac{3\, \alpha\, \tilde{r} \,\rmin \, \qmin}{32 \, \T} \frac{c' (1 - \alpha)}{d \kappa} \geq \frac{c' \, \tau \, \rmin }{11   d   \kappa   T}, 
\end{equation}
for some universal constant $c' > 0$, where we bound $\Gap(\barP)$ using~\eqref{eq:gap-barP-mix} and bound $\Gap(P_{i, j})$ by \cref{lm:rwm}. The claim then follows from~\eqref{eq:coro-rwm-1}.  
\end{proof}

\begin{remark}\label{rmk:order}
Assume $\tau = \alpha (1-\alpha) \, \qmin \, \tilde{r} = \Theta(1)$,  $\rmin = \Theta(\T^{-1})$,  and let $d, L D^2 \rightarrow \infty$.  
Using the expression for $\T$ given in~\eqref{eq:choice-beta} and \cref{coro:main-rwm}, we find that in order to reach $\epsilon$-accuracy in TV distance, the number of simulated tempering iterations needed is 
\begin{align}\label{eq:comp-rwm}
O \left( \frac{ d^2 \kappa^3  \log^2 (LD) }{\wmin^2}   \log \frac{\tilde{\eta}}{2\epsilon} \right) \text{ when $Q_i$ is the RWM proposal.} 
\end{align}
By \cref{th:main}, assuming $\log (T \eta / \epsilon) = O(d)$, the number of iterations needed is 
\begin{equation}\label{eq:comp-mala}
    O \left(  \frac{ d^3 (d + m D^2) \kappa^6   \log^4 (LD)  }{\wmin^4}  \log \frac{2\eta}{\epsilon} \right) \text{ when $Q_i$ is the MALA proposal.}
\end{equation}
The bound for the MALA proposal is roughly the square of that for the RWM proposal mainly for two reasons. 
First, the approach we take in \cref{sec:conduct-restrict} to bounding the $s$-conductance is likely suboptimal and underestimates the $s$-conductance of $P_{i, j}$ with the MALA proposal. Second, when estimating the mixing time, we use the fact that the TV distance decays exponentially at rate  $\Phi_s^2(P)$ or $\Gap(P)$. Since both the decomposition bound for $\Phi_s(P)$ given in~\cref{th:conduct-decomp}  and that for $\Gap(P)$ given in~\eqref{eq:gap-decomp} are proportional to $\Gap(\barP)$, the $s$-conductance approach effectively squares the dependence on $\Gap(\barP)$, and by Remark~\ref{rmk:project-gap},  $1/\Gap(\barP)$ typically has order $\Theta(\T^2) = \Theta( d \kappa^2 \log^2(LD) )$. 
\end{remark}

\subsection{Spectral Gap of the Projected Chain}\label{sec:gap-projected}

Consider the projected chain $\barP$. By \cref{def:decomp}, for our analysis the transition matrix of $\barP$ can be expressed by 
\begin{equation}
    \barP( (i, j), \, (i', j') ) 
    =  \int_{\bbR^d}   \pi_{i,j}(x) P( (i, j, x), \, (i', j', x ) ) \, \dd x, \text{ if } (i,j) \neq (i', j'), 
\end{equation}
where we have used that in each transition governed by $P$, $x$ cannot be changed if either $i$ or $j$ is modified. 
The stationary distribution of $\barP$ is given by $\pi(i, j) =  r_i w_j$. 
To bound $\Gap(\barP)$, we  use the standard canonical path method~\citep{sinclair1992improved}.  The paths are constructed based on the following intuition: if $\beta_1$ is sufficiently small, the component distributions $(\pi_{1, j})_{j \in [\K]}$ overlap substantially, allowing $\barP$  to easily move between $(1, j)$ and $(1, j')$ for any $j \neq j'$; at lower temperatures, for example at $\beta_{\T} = 1$, the chain $\barP$ can move from $(\T, j)$ to $(\T, j')$ by following   the path 
\begin{equation}
    (\T, j) \rightarrow (\T-1, j) \rightarrow \cdots \rightarrow (1, j) \rightarrow (1, j') \rightarrow (2, j') \rightarrow \cdots \rightarrow (\T, j'). 
\end{equation}
Since this argument is commonly used in the simulated tempering literature, we give the details in \cref{lm:canonical-path}.  
But different from existing results~\citep{ge2018simulated, garg2025restricted},  we bound the transition probabilities $\barP( (1, j), \, (1, j'))$ using the Hellinger affinity between two adjacent component distributions, as we describe in the following lemma. 

\begin{lemma}\label{lm:project-chain}
Consider the projected chain $\barP$ defined above. 
Let $\qmin$ be as given in~\eqref{ass:Q} and $\rmin, \tilde{r}$ be as given in~\eqref{eq:def-complex-parameters}.  Define 
\begin{align}  
    \Delta   = \frac{1}{2} \max_{i \in [\T - 1], j \in [\K]} \KL( \Pi_{i, j} \,\|\, \Pi_{ i + 1 , j } ), \quad 
    H  = \min_{j, j' \in [\K]} \int_{\bbR^d} \sqrt{ \pi_{1, j}(x) \pi_{1, j'}(x) } \,  \dd x.  \label{eq:def-H}
\end{align}
Then, 
\begin{equation}
    \Gap(\barP) \geq \frac{ \alpha \, \tilde{r} \, \rmin \, \qmin  }{4\T} \min \left\{ 1 - \sqrt{\Delta}, \, H^2 \right\}. 
\end{equation}
\end{lemma}
\begin{proof}
By the standard canonical path argument which we detail in \cref{lm:canonical-path},  we can bound $\Gap(\barP)$ by $\Gap(\barP)  \geq (2 \T)^{-1}  \min\{ \Lambda_1, \Lambda_2 \},$  
where 
\begin{align}
    \Lambda_1 = \min_{i \in [\T - 1], j \in [\K]}  r_i \, \barP( (i, j), \, (i+1, j) ),  \quad 
    \Lambda_2 = \min_{j, j' \in [\K]} \frac{ \rmin }{w_{j'} }  \barP( (1, j), \, (1, j') ).  
\end{align}

For $\Lambda_1$, recall that the proposal scheme $Q$ given in~\eqref{eq:def-Q-joint} proposes a temperature change from level $i$ to level $i+1$ with probability $(\alpha / 2) q(i, i+1)$. Hence, using~\eqref{ass:Q}, 
\begin{align*}
 \barP( (i, j), \, (i+1, j) ) &= \frac{\alpha \, q(i, i+1)}{2} \int \pi_{i, j}(x)  a_{P}( (i, j, x), \, (i+1, j, x) ) \dd x \\
 &= \frac{\alpha \, q(i, i+1)}{2  } \int \pi_{i, j}(x)
 \min\left\{ 1, \,  \frac{ r_{i+1}  \pi_{i + 1, j}(x) }{ r_{i}  \pi_{i, j}(x) } \right\} \dd x \\  
  &\geq  \frac{\alpha \, \tilde{r} \, \qmin  }{2   } \int  \min\left\{   \pi_{i, j}(x), \,    \pi_{i + 1, j}(x)   \right\} \dd x. 
\end{align*} 
Since $\int  \min\left\{   \pi_{i, j}(x), \,    \pi_{i + 1, j}(x)   \right\} \dd x = 1 - \|  \Pi_{i, j} - \Pi_{ i + 1 , j } \|_{\TV}$, applying Pinsker's inequality yields
\begin{equation}\label{eq:Lambda-1}
  \Lambda_1 \geq \rmin \barP( (i, j), \, (i+1, j) )  \geq \frac{1}{2}  \alpha \, \tilde{r}  \, \rmin \, \qmin  \left(1 - \sqrt{ \frac{1}{2}  \KL( \Pi_{i, j} \,\|\, \Pi_{ i + 1 , j } )} \right). 
\end{equation}

For $\Lambda_2$, recall $Q$ proposes changing $j$ with probability $1/2$, and $j'$ is drawn from the full conditional distribution. Using $ \pi_{i, x}(j')  = w_{j'} \pi_{i, j'}(x) / \pi_i(x)$ where $\pi_i(x) = \sum_j w_j \pi_{i, j}(x)$,   we get 
\begin{align*}
 \barP( (1, j), \, (1, j') ) &= \frac{1}{2} \int \pi_{i, j}(x)  \pi_{i, x}(j') \dd x \\ 
 &=  \frac{w_{j'}}{2} \int \frac{ \pi_{i, j}(x) \pi_{i, j'}(x) }{ \pi_i(x) } \dd x \\
 & \geq \frac{w_{j'}}{2} \left\{ \int \sqrt{ \pi_{i, j}(x) \pi_{i, j'}(x) } \dd x   \right\}^2, 
\end{align*}  
where the last step follows from the Cauchy--Schwarz inequality. Therefore,
\begin{equation}\label{eq:Lambda-2}
\Lambda_2 = \min_{j, j' \in [\K]} \frac{\rmin}{w_{j'} }  \barP( (1, j), \, (1, j') ) \geq \frac{1 }{ 2 }  \rmin  H^2. 
\end{equation}  

Combining~\cref{lm:canonical-path},~\eqref{eq:Lambda-1} and~\eqref{eq:Lambda-2}, we get 
\begin{equation}
    \Gap(\barP)  \geq \frac{1}{2 \T}  \min\{ \Lambda_1, \Lambda_2 \}
    \geq \frac{\rmin }{4 \T} \min\{ \alpha \, \tilde{r}  \, \, \qmin  (1 - \sqrt{ \Delta } ), \;     H^2  \},
\end{equation}
from which the claim follows. 
\end{proof}

\begin{remark}\label{rmk:project-gap}
\cref{lm:project-chain} implies that if $\alpha, \qmin, \tilde{r}, 1 - \sqrt{\Delta}, H^2$ are all bounded away from zero and $\rmin = \Theta(\T^{-1})$, then $\Gap(\barP) = \Omega( \T^{-2})$.
This spectral gap bound in general cannot be improved for any reversible simulated tempering algorithm, since the spectral gap of a simple lazy random walk on $[\T]$ is $O(T^{-2})$~\citep[Chap. 2.2]{montenegro2006mathematical}. 
\end{remark}

Next, we show that if $(\beta_i)_{i \in [\T]}$ are chosen properly, we can bound $\min \{ 1 - \sqrt{\Delta}, \, H^2  \}$ from below by a universal constant. 
The quantity $H$, which denotes the minimum Hellinger affinity between two component distributions with inverse temperature $\beta_1$, can be bounded by only using the smoothness parameter $L$ and the maximum separation between the modes $\mu_1, \dots, \mu_{\K}$. 
Bounding $\Delta$ is more involved and requires an integral representation of  $\KL(\Pi_{i, j} \,\|\, \Pi_{i+1, j})$ using the Fisher information for the inverse temperature parameter~\citep[Chap. 3.5]{amari2000methods}. 

\begin{lemma}\label{lm:H-delta}
Consider $\Delta, H$   defined in~\eqref{eq:def-H}, and let $D$ be as given in~\eqref{eq:def-complex-parameters}. Then, 
\begin{equation}
  H \geq e^{ - \beta_1 L D^2 / 2},  \quad \Delta \leq \frac{d \kappa^2}{4} \left( \max_{i \in [\T - 1]} \frac{\beta_{i+1}}{\beta_i} - 1 \right)^2. 
\end{equation}
Hence, assuming $L D^2 > 1/4$ and $(\beta_i)_{i \in [\T]}$ satisfies 
\begin{equation} \label{eq:cond-beta}
    \beta_1 \leq \frac{1}{4 L D^2}, \quad \frac{\beta_{i+1}}{\beta_i} \leq 1 + \frac{1}{\kappa \sqrt{d} } \text{ for } i \in [\T - 1], 
\end{equation}
we have $\min\{1 - \sqrt{\Delta}, \, H^2 \} \geq \min\{ 1 - 1/4, \, e^{-1/4} \} \geq 3/4. $ 
\end{lemma}

\begin{proof}
Consider $H$ first, and fix any $j \neq j'$. 
Since $f$ is $L$-smooth, by \cref{lm:smooth}, 
\begin{equation}
    \frac{1}{2}\left[f(y) + f(z)\right] \leq  f\left( \frac{y + z}{2} \right) + \frac{L}{8} \|y - z \|^2, 
\end{equation}
for any $y, z \in \bbR^d$.  
Applying this inequality with $y = x - \mu_j$ and $z = x - \mu_{j'}$, we get 
\begin{align*}
\int_{\bbR^d} \sqrt{ \pi_{1, j}(x) \pi_{1, j'}(x) } \,  \dd x  
& = \frac{1}{C_1} \int_{\bbR^d} \exp\left\{ -\frac{\beta_1}{2}[ f(x- \mu_j) + f(x - \mu_{j'})] \right\}  \dd x \\
&\geq \frac{1}{C_1} \int_{\bbR^d} \exp\left\{ - \beta_1 \left[ f(x - \bar{\mu}) + \frac{L}{8} \| \mu_j - \mu_{j'} \|^2  \right] \right\}  \dd x \\
& = \exp( - \beta_1 L \| \mu_j - \mu_{j'} \|^2 / 8) , 
\end{align*}
where $\bar{\mu} = (\mu_j + \mu_{j'})/2$ and 
$C_1$ is the normalizing constant defined in~\eqref{eq:def-joint-pi}. Since $j, j'$ are arbitrary and $ \| \mu_j - \mu_{j'} \|^2 \leq 4 D^2$, this proves the claimed bound for $H$. 

Next, consider the quantity $\Delta$. 
Let $\Pi_{\beta, 0}$ denote the probability distribution with density $\pi_{\beta, 0}(x) \propto e^{- \beta f(x)}$. 
Using an integral representation of KL divergence for exponential families which we give in \cref{lm:KL},  we get 
\begin{equation}\label{eq:KL-fisher}
     \KL( \Pi_{i, j} \,\|\, \Pi_{i+1, j} )    = \int_{\beta_i}^{\beta_{i+1}} (\beta_{i + 1} - \beta)  V_{\Pi_{\beta, 0}} (f) \, \dd \beta. 
\end{equation} 

To bound the variance  $V_{\Pi_{\beta, 0}} (f) $, we first note that since $f$ is minimized at $0$, we have  $\nabla f (0) = 0$. It then follows from \cref{lm:smooth} that $\|  \nabla f (x)  \|^2 \leq L^2 \| x \|^2$. 
Applying  the Poincar\'{e} inequality and~\eqref{eq:durmus}  in \cref{lm:log-concave}, we obtain 
\begin{equation}
     V_{\Pi_{\beta, 0}} (f)  \leq  \frac{1}{\beta m} \Pi_{\beta, 0}( \|\nabla f \|^2 ) \leq \frac{ L^2 }{\beta m}  \int_{\bbR^d} \| x \|^2 \Pi_{\beta, 0} (\dd x)  \leq \frac{d L^2}{\beta^2 m^2} = \dfrac{d \kappa^2}{ \beta^2 }. 
\end{equation}
Hence, $ V_{\Pi_{\beta, 0}} (f)  \leq d \kappa^2 / \beta_i^2$ for $\beta \in [\beta_i, \beta_{i + 1}]$, and it follows from~\eqref{eq:KL-fisher} that 
\begin{equation}
      \KL( \Pi_{i, j} \,\|\, \Pi_{i+1, j} )  \leq \frac{d \kappa^2}{\beta_i^2} \int_{\beta_i}^{\beta_{i+1}}  ( \beta_{i + 1} - \beta ) \dd \beta = \frac{d \kappa^2}{2 \beta_i^2} (\beta_{i+1} - \beta_i)^2,  
\end{equation} 
which concludes the proof. 
\end{proof}

We compare our choice of $(\beta_i)_{i \in [\T]}$ with those in the simulated tempering literature. 
Assume that $\beta_{i + 1} / \beta_i \coloneqq 1 + \rho$ is a constant. Ideally, one wants to achieve polynomial mixing with as few temperatures as possible, which means that larger values of $\beta_1$ and $\rho$ are more desirable.  
In~\citet{atchade2011towards},  optimal scaling techniques were used to show that for certain target distributions including standard Gaussian distribution, the geometric spacing with $\rho = 1/\sqrt{d}$ is asymptotically optimal. Although their analysis is asymptotic and does not directly cover mixture targets, this result already suggests that the $1/\sqrt{d}$-dependence is likely unavoidable. 
The work of~\citet{garg2025restricted} considered only mixtures of Gaussian distributions and assumed fixed $d$, and their conditions on the temperature ladder are essentially the same as ours. 
In~\citet{woodard2009conditions}, polynomial mixing of simulated tempering was obtained when the target is a mixture of two isotropic Gaussian distributions with means $\mu_1 = \one, \mu_2 = -\one$, and $\beta_i = d^{-(d + 1 - i)/d}$ for $i = 1, \dots, d + 1$. Since in this case $D^2   = d$, their choice of $\beta_1 = d^{-1}$ coincides with ours. However, their temperature spacing is sub-optimal: as $d \rightarrow \infty$, their construction of $(\beta_i)_{i \in [\T]}$ implies $ \rho \approx (\log d )/ d$. 
In~\citet{ge2018simulated}, the authors considered the same target distribution~\eqref{eq:target-pi-star} and assumed $ \beta_1 \lesssim  1/ (d \kappa  D^2)$  and $\rho \leq  1/ [\kappa d \log (\kappa + 1) ]$. Both conditions are more restrictive than ours.

\subsection{Conductance of the Restricted Chains}\label{sec:conduct-restrict}

Consider the restricted chain $P_{i, j}$. By \cref{def:decomp},  $P_{i, j}$ is essentially the Metropolis--Hastings chain on $\bbR^d$ with stationary density $\pi_{i, j}$ and proposal $Q$ defined in~\eqref{eq:def-Q-joint} restricted to only updates of $x$. Explicitly, 
\begin{equation}
    P_{i, j}(x, A) = \frac{1-\alpha}{2} \int_A  \min\left\{1, \, \frac{ \pi_{i, j}(x') q_i(x', x)}{ \pi_{i, j}(x) q_i(x, x')} \right\} Q_i(x, \dd x'), \quad \text{ if } x \notin A. 
\end{equation}
It follows that if $P'_{i, j}$ is the transition kernel of a Metropolis--Hastings chain with stationary density $\pi_{i, j}$ and proposal $Q_i$, then 
\begin{equation}
   \Phi_s(P_{i, j}) = \frac{1-\alpha}{2}\Phi_s(P'_{i, j}), \quad \Gap(P_{i, j} ) = \frac{1-\alpha}{2}\Gap(P'_{i, j}).
\end{equation}

When $Q_i$ is the RWM proposal, we can simply apply the recent result of~\citet{andrieu2024explicit} to get the following spectral gap  and conductance bounds. 

\begin{lemma}\label{lm:rwm}
Consider the restricted Metropolis--Hastings chain $P_{i, j}$ defined above with  proposal 
$Q_i(x, \cdot) = N(x, (2h \beta_i^{-1}) I_d)$.
If $h = 1/ (L d)$, then for any $i \in [\T], \,  j \in [\K]$, 
\begin{equation}
   \Phi_0(P_{i, j}) \geq \frac{c (1 - \alpha)}{\sqrt{d \kappa}}, \quad  \Gap(P_{i, j}) \geq  \frac{c (1 - \alpha)}{ d \kappa }, 
\end{equation}
where $c$ is some universal constant. 
\end{lemma}
\begin{proof}
    Since  $\pi_{i, j}(x) \propto e^{ -\beta_i f(x - \mu_j)}$, its potential is $(\beta_i L)$-smooth and $(\beta_i m)$-strongly convex. The claim then follows from Corollary 35 of~\citet{andrieu2024explicit}. 
\end{proof}

When $Q_i$ is the MALA proposal, the analysis becomes more challenging. Recall that when we constructed the auxiliary chains $P$ (and also $\tilde{P}$), the proposal scheme for updating $x$ is assumed to be the same as that used by the proposal kernel $Q^*$ of the actual simulated tempering algorithm. Hence, the MALA proposal is computed by assuming the target distribution is $\pi^*_i(x)$, which has potential $\beta_i U(x)$ with $U$ given in~\eqref{eq:target-pi-star}. However, the stationary distribution of $P_{i, j}$ is $\pi_{i, j}(x)$, which differs from $\pi^*_i(x)$. 
Consequently, though recent studies have provided sharp bounds on the $s$-conductance of MALA for log-concave targets~\citep{chewi2021optimal, wu2022minimax}, such results cannot be directly used here. In the next lemma, we take the approach of~\citet{dwivedi2019log} to bounding the $s$-conductance of $P_{i, j}$, which is known to be suboptimal for a true MALA chain targeting some log-concave distribution. 
The calculation is more involved due to the mismatch between $\pi_{i, j}(x)$ and $\pi^*_i(x)$, and most technical details are deferred to Appendix~\ref{sec:supp-Pij}.

\begin{lemma}\label{lm:mala}
Consider the restricted Metropolis--Hastings chain $P_{i, j}$  with proposal  $Q_i(x, \cdot) = N(x - h \nabla U(x), (2h/\beta_i) I_d)$.
For any $s \in [0, 1/2)$,  let
\begin{equation}
    h = \frac{c  }{   L^2 (D + \cR_s)^2  d}, \text{ where } \cR_s = 2 \sqrt{ \frac{d \vee (2 \log s^{-1})}{m} }.  
\end{equation}
If $c \leq 0.01$,  then  $\Phi_s(P_{i, j}) \geq c' \sqrt{h m}$ for some universal constant $c' > 0$. 
\end{lemma}

\begin{proof} 
Fix $i \in [\T], j \in [\K]$.  By \cref{lm:log-concave}, we have $\Pi_{i, j}(\cB_s^i) \geq 1 - s$, where 
\begin{equation}
    \cB_{s}^i = \left\{ x \in \bbR^d \colon \| x \| \leq  \cR_s/ \sqrt{\beta_i }     \right\}. 
\end{equation} 
By \cref{lm:s-conductance-isoperi}, since $\pi_{i, j}$ is $(\beta_i m)$-strongly convex, if we can show that 
\begin{equation}\label{eq:goal-Pij}
    \| P_{i, j} (x, \cdot) - P_{i, j}(y, \cdot) \|_{\TV} \leq 3/4, \text{ whenever } x, y \in \cB_s^i, \text{ and } \| x - y\| \leq  \sqrt{ \frac{h}{8 \beta_i}}, 
\end{equation}
then $\Phi_s(P_{i, j}) \geq c' \sqrt{h m}$  for some universal constant $c' > 0$. 

To prove~\eqref{eq:goal-Pij}, we begin by applying  Pinsker's inequality and the well-known formula for  the KL divergence between two Gaussian distributions to get 
\begin{align}
    \| Q_i(x, \cdot) - Q_i(y, \cdot) \|_{\TV}    &\leq \frac{\sqrt{\beta_i}}{2 \sqrt{2 h}} \left\|x - h \nabla U(x) - (y - h \nabla U(y) )  \right\|,  \quad \forall \, x, y \in \bbR^d. 
\end{align}
By  \cref{lm:U-bound}, $\| \nabla U(x) - \nabla f(x)\| \leq LD$ for any $x$.
It then follows from   \cref{lm:mala-mean-bound} that 
\begin{align}
     \left\|x - h \nabla U(x) - (y - h \nabla U(y))   \right\|   
     & \leq \left\|x - h \nabla f(x) - (y - h \nabla f(y))   \right\|  + 2 h LD 
     \leq  \| x - y\| + 2 h LD. 
\end{align} 
Using the assumption on $h$ and that $\beta_i \leq 1$, we find that  
\begin{align}
     \| Q_i(x, \cdot) - Q_i(y, \cdot) \|_{\TV} \leq \frac{\sqrt{\beta_i} \|x - y \|}{2 \sqrt{2 h}} + \frac{ \sqrt{h} L D}{\sqrt{2}} \leq \frac{\sqrt{\beta_i} \|x - y \|}{2 \sqrt{2 h}} + \frac{ \sqrt{c} }{\sqrt{2 }}. 
\end{align}
An application of the triangle inequality thus yields, when $x, y \in \cB_s^i$ and $\| x - y\| \leq \sqrt{h / (8\beta_i)}$, 
\begin{align}\label{eq:tv-pij}
     \| P_{i, j} (x, \cdot) - P_{i, j}(y, \cdot) \|_{\TV} \leq  \frac{1}{8} + \frac{ \sqrt{c } }{\sqrt{2 }} +     2 \sup_{x \in \cB_s^i} \| P_{i, j}(x, \cdot) - Q_i(x, \cdot) \|_{\TV}. 
\end{align}

We now fix any  $x \in \cB_s^i$ and bound $\| P_{i, j}(x, \cdot) - Q_i(x, \cdot) \|_{\TV} $. It is easy to see that it equals the probability that the proposal $Y \sim Q_i(x, \cdot) $ is rejected. So we introduce the notation $\bar{a}(x)$, which denotes the average acceptance probability at $x$ and can be expressed by 
\begin{equation}\label{eq:def-bar-acc}
\bar{a}(x) = \int_{\bbR^d} \min\left\{1, \frac{e^{-\beta_i f(y(z) - \mu_j)} }{e^{-\beta_i f(x - \mu_j)}} 
\frac{ q_i(y(z), x)}{q_i(x, y(z))} \right\} \cN_d(\dd z),  
\end{equation} 
where $\cN_d$ denotes the $d$-dimensional standard Gaussian measure, 
$q_i(x, y)$ denotes the density of $Y \sim N(x - h \nabla U(x), (2h/\beta_i) I_d)$, 
and
\begin{equation}
     y(z) = x - h \nabla U(x) + \sqrt{2h / \beta_i } \, z.
\end{equation} 

We apply~\cref{lm:acc-mala} with $c_h = c/d$  to get 
\begin{align}
  \log  \frac{e^{-\beta_i f(y(z) - \mu_j)} }{e^{-\beta_i f(x - \mu_j)}} & \frac{ q_i(y(z), x)}{q_i(x, y(z))} 
  \geq  - b_0 - b_1 \|z\| - b_2 \|z\|^2,  \\ 
    \text{ where }   b_0 &= \left( 1 + \frac{hL}{2} \right)\left( 2 + \frac{hL}{2} \right)   \frac{c }{d} \leq  \left(2  + \frac{7}{4}h L\right) c, \\
    b_1 &= (4 + 5 hL + h^2L^2) \sqrt{ \frac{ c  }{2 d} } \leq (4 + 6hL) \sqrt{ \frac{ c  }{2 d} }, \\
    b_2 &= h L   \left( 2 + \frac{hL}{2} \right) \leq \frac{c}{4 d} \left( 2 + hL \right), 
\end{align} 
where to derive the bounds for $b_0, b_1, b_2$, we have used 
\begin{equation}\label{eq:LD2} 
 h L \leq \frac{c }{L (D + \cR_s)^2 d} \leq \frac{c }{L \cR_s^2 d} \leq \frac{c }{4 L d^2/ m} \leq  \frac{c }{4d} < 1, 
\end{equation}
since $\cR_s^2 \geq 4 d / m$  and $L \geq m$.  By Jensen's inequality, 
\begin{align*}
   \bar{a}(x) & \geq \int_{\bbR^d} e^{  - b_0 - b_1 \|z\| - b_2 \|z\|^2 }   \cN_d(\dd z)  
     \geq   \exp \left(- \int_{\bbR^d}   [ b_0 + b_1 \| z \| + b_2 \|z\|^2 ]\, \cN_d(\dd z) \right)   \\ 
   & \geq    \exp \left(- b_0 - b_1 \sqrt{d} - b_2 d \right)  
    \geq   \exp \left\{ - \left( \frac{5}{2} + 2 hL\right) c  - (4 + 6 hL) \sqrt{\frac{c}{2}}   \right\}. 
\end{align*}  

Using $c  \leq 0.01$ and $hL \leq c / 4$, 
one can numerically verify that the right-hand side is greater than $0.73$. Hence, 
$$\| P_{i, j}(x, \cdot) - Q_i(x, \cdot) \|_{\TV} \leq 1 - \bar{a}(x) \leq 0.27.$$
Since this holds for any $x \in \cB_s^i$, another numerical calculation with $c\leq 0.01$ verifies that the right-hand side of~\eqref{eq:tv-pij}  is less than $0.75$, which establishes~\eqref{eq:goal-Pij} and completes the proof.  
\end{proof}

\section{Necessary Conditions on the   Temperature Ladder} \label{sec:opt-beta}

In this section, we further analyze the choice of the temperature ladder $(\beta_i)_{i \in [\T]}$. 
In~\eqref{eq:opt-choice-beta} and \cref{lm:H-delta}, we have stated our conditions on $\beta_1$ and $\beta_{i + 1}/\beta_i$ for $i \in [\T - 1]$, which jointly also determine the number of temperature levels $\T$. 
We now show that these conditions are asymptotically optimal (up to logarithmic factors) in terms of the dependence on the dimension $d$ and maximum mode displacement $D$. 

We simply fix the target distribution to be a mixture of two equally weighted Gaussian distributions with covariance matrix $I_d$: 
\begin{equation}\label{eq:counter-pi-star}
    \pi^*(x) = \frac{1}{2 (2 \pi)^{d / 2}} \left( e^{ - \|x - \mu_1\|^2 / 2} + e^{ - \|x - \mu_2\|^2 / 2} \right). 
\end{equation} 

We first consider the choice of $\beta_1$ and assume that the RWM proposal is used. 

\begin{proposition}\label{prop:opt-beta1}
Let $\pi^*$ be given by~\eqref{eq:counter-pi-star} with  $\mu_1 = (D, 0, \dots, 0)$ and $\mu_2 = -  (D, 0, \dots, 0)$ for some $D > 0$, which implies $\max_j \|\mu_j\| = D$. 
Consider the simulated tempering chain $P^*$ with proposal $Q_i(x, \cdot) = N(x, \, (2h \beta_i^{-1}) I_d )$ for some $h > 0$. For any $s \in [0, 1/2)$, 
\begin{equation}
\Phi_s(P^*) \leq \frac{4}{1 - 2s} \left( \frac{2}{2 + h} \right)^{d/2} \exp\left( - \frac{\beta_1 D^2}{2 + h} \right). 
\end{equation} 
\end{proposition} 

\begin{proof}
Define the set 
\begin{equation}
    A = \{ (i, x): i \in [\T], x = (x_1, \dots, x_d) \in \bbR^d,  x_1 > 0\}. 
\end{equation} 
By symmetry, $\Pi^*(A) = 1/2$. Consider the probability flow from $A$ to $A^c$ of the Markov chain $P^*$. Since for $x \neq y$, the transition density  from $(i, x)$ to $(j, y)$ is nonzero only when $j = i$, we have 
\begin{align}
P^*(A, A^c) &= \sum_{i\in [\T]} r_i P^*_i(A, A^c) \leq \max_{i \in [\T]} P^*_i(A, A^c),  \label{eq:beta1-temp1} \\
\text{ where }
   P^*_i(A, A^c) &= \int_{x\colon  x_1 > 0} \int_{y\colon  y_1 < 0}  q_i(x, y) \min \left\{\pi^*_i(x), \pi^*_i(y) \right\} \dd x \, \dd y. \label{eq:opt-beta1-1} 
\end{align} 
By \cref{lm:ge-compare}, $\pi^*_i(x) \leq 2 \pi_i(x)$, where   $\pi_i(x) = \sum_{j \in [\K]} w_j \pi_{i, j}(x)$. 
When $x_1 > 0$, we have $\pi_{i, 1}(x) \geq \pi_{i, 2}(x)$ since $\|x - \mu_1 \| \leq \|x - \mu_2\|$. Hence, we have $\pi_i(x) \leq \pi_{i, 1}(x)$, and similarly, $\pi_i(y) \leq  \pi_{i, 2}(y)$. It follows that 
\begin{equation}
\min \left\{\pi^*_i(x), \pi^*_i(y) \right\} \leq 2
 \min \left\{\pi_i(x), \pi_i(y) \right\} \leq 2 \sqrt{ \pi_i(x)  \pi_i(y) } \leq 2 \sqrt{ \pi_{i, 1}(x)  \pi_{i, 2}(y) }. 
\end{equation}
Using the expression given in~\eqref{eq:opt-beta1-1},  we thus find that 
\begin{equation}
     P^*_i(A, A^c)    \leq \int_{\bbR^d \times \bbR^d}  q_i(x, y) \sqrt{ \pi_{i, 1}(x) \pi_{i, 2}(y) } \, \dd x \, \dd y    = 2  \left( \frac{2}{2 + h} \right)^{d/2} \exp\left( - \frac{\beta_i D^2}{2 + h} \right), 
\end{equation}
where the second equality is obtained via routine calculations with Gaussian densities. 
By the definition of $s$-conductance, $\Phi_s(P^*) \leq P^*(A, A^c) / (\Pi^*(A) - s)$ for any $s \in [0, 1/2)$. 
Applying~\eqref{eq:beta1-temp1} and using $\beta_1 = \min_{i \in [\T]} \beta_i$ concludes the proof. 
\end{proof}

Next, we determine the necessary condition needed on the ratio $\beta_{i + 1} / \beta_i$. 

\begin{proposition}\label{prop:opt-beta-ratio}
Let $\pi^*$ be given by~\eqref{eq:counter-pi-star}, and consider the simulated tempering chain $P^*$ with $r_i = 1 / \T$ for each $i \in [\T]$. 
Define $\rho_i = \beta_{i+1}/\beta_i - 1$ for each $i \in [\T] - 1$. 
For any $s \in [0, (4 \T)^{-1})$,  
\begin{equation}\label{eq:def-F}
    \Phi_s(P^*)  \leq 16 \min_{i \in [\T - 1]}   F(\rho_i), \text{ where } 
    F(\rho) = \left(\frac{ 2 \sqrt{ 1 + \rho  }  }{2 + \rho   } \right)^{d/2}. 
\end{equation} 
The function $F(\rho)$ is monotone decreasing on $[0, \infty)$, and for $\rho \in [0, 1/2]$, $F(\rho) \leq e^{- \rho^2 d / 48}$. 
\end{proposition}

\begin{proof}
Consider the auxiliary chain $\tilde{P}$ defined in \cref{sec:first-aux-chain}. For $A_i = \{i\} \times \bbR^d$, we have 
\begin{align}
    \tilde{P}(A_i, A_{i'}) &=  \int_{\bbR^d} \pi(i, x) q(i, i') \min\left\{ 1, \, \frac{\pi(i', x)}{\pi(i, x) }\right\} \dd x \\ 
    & = q(i, i')  \int_{\bbR^d}    \min\left\{ r_i \pi_i(x), \, r_{i'} \pi_{i'}(x) \right\} \dd x \\
    & = \frac{ q(i, i') }{\T}    \int_{\bbR^d}  \sqrt{\pi_i(x) \pi_{i'}(x) } \dd x. \label{eq:ratio-1}
\end{align}
Since $\pi_i(x) = 0.5 \pi_{i, 1}(x) + 0.5 \pi_{i, 2}(x)$ and $\sqrt{a_1 + a_2} \leq \sqrt{a_1} + \sqrt{a_2}$, we have  
\begin{equation}\label{eq:ratio-2}
     \sqrt{\pi_i(x) \pi_{i'}(x) } \leq \frac{1}{2} \sum_{j =1 }^2 \sum_{j' = 1}^2 \sqrt{ \pi_{i, j}(x) \pi_{i', j'}(x) }.  
\end{equation}

A routine calculation yields that 
\begin{equation}\label{eq:ratio-3}
    \int_{\bbR^d} \sqrt{ \pi_{i, j}(x) \pi_{i', j'}(x) } \dd x 
    = \left(\frac{ 2 \sqrt{\beta_i  \beta_{i'}}  }{\beta_i + \beta_{i'}} \right)^{d/2}
    \exp\left\{ - \frac{ \beta_i \beta_{i'} \|\mu_j - \mu_{j'}\|^2 }{4 (\beta_i + \beta_{i'})} \right\}. 
\end{equation}
We simply bound the exponential term by $1$. 
Substituting~\eqref{eq:ratio-2} and~\eqref{eq:ratio-3} into~\eqref{eq:ratio-1} yields, for $i' > i$, 
\begin{align}\label{eq:opt-ratio-temp1}
    \tilde{P}(A_i, A_{i'}) &\leq   \frac{2 q(i, i')}{\T}    \left(\frac{ 2 \sqrt{\beta_i  \beta_{i'}}  }{\beta_i + \beta_{i'}} \right)^{d/2}
    =  \frac{2 q(i, i')}{\T} F \left( \frac{\beta_{i'}}{\beta_i} - 1 \right),
\end{align}
where $F$ is given by~\eqref{eq:def-F}. 
It is straightforward to verify that the function $F(\rho)$ is monotone increasing on $(-\infty, 0]$ and monotone decreasing on $[0, \infty)$.  
Hence, letting $B_i = \cup_{k \leq i} A_k$ and $\rho_i = \beta_{i+1}/\beta_i - 1$,  we find that 
\begin{align}
    \tilde{P}( B_i, B_i^c) = \sum_{k \leq i} \sum_{\ell \geq i + 1} \tilde{P}(A_k, A_\ell) 
    \leq  F(\rho_i) \sum_{k \leq i} \sum_{\ell \geq i + 1} \frac{2 q(k, \ell)}{\T} 
    \leq \frac{2 i}{\T} F(\rho_i). 
\end{align} 
Since $\tilde{\Pi}(B_i) = \sum_{k \leq i} r_k = i / \T$, we can apply \cref{lm:compare-st-one} to get 
\begin{equation}
    \Phi_s(P^*) \leq 4 \Phi_{2s} (\tilde{P}) \leq   \frac{4 \, \tilde{P}( B_i, B_i^c) }{ \tilde{\Pi}(B_i) - 2s } \leq 16 F(\rho_i), 
\end{equation}
provided that $s \leq i / (4 \T)$. This yields the claimed $s$-conductance bound. To prove the last claim about $F(\rho)$,  we use $\rho - \rho^2 /2 \leq \log(1 + \rho) \leq \rho - \rho^2/2 + \rho^3/3$ for $\rho \geq 0$ to get
 \begin{align}
     \log \frac{2 \sqrt{1 + \rho}}{2 + \rho}
     = \frac{1}{2}\log (1 + \rho) - \log \left(1 + \frac{\rho}{2} \right)  
     \leq - \frac{\rho^2}{8} + \frac{\rho^3}{6}. 
 \end{align} 
A simple calculation shows that the right-hand side is bounded by $-\rho^2 / 24$ on $[0, 1/2]$. 
\end{proof}

Combining \cref{prop:opt-beta1} and \cref{prop:opt-beta-ratio}, we obtain the following necessary conditions on the temperature ladder for polynomial complexity of the simulated tempering algorithm with respect to both $d$ and $D$. They show that   our choice of the   temperature ladder given in~\eqref{eq:opt-choice-beta}  achieves the asymptotic optimal dependence on $d, D$ up to logarithmic factors. 

\begin{corollary}\label{coro:necessary}
    Let $\pi^*$ be given by~\eqref{eq:counter-pi-star} with  $\mu_1 = (D, 0, \dots, 0)$ and $\mu_2 = -  (D, 0, \dots, 0)$ for some $D > 0$. Consider the simulated tempering chain $P^*$ with $r_i = 1 / \T$ and  proposal $Q_i(x, \cdot) = N(x, \, (2h \beta_i^{-1}) I_d )$ for some $h > 0$. Let $s \leq (4 \T)^{-1}$. If $\Phi_s(P^*) = \Omega ( d^{-c} D^{-c} )$ for some fixed $c \geq 0$ as $d, D \rightarrow \infty$,  then
    \begin{equation}
        \beta_1 = O \left( \frac{\log D}{D^2} \right), \quad \frac{\beta_{i+1}}{\beta_i} = 1 + O \left( \frac{\sqrt{\log d} }{\sqrt{d}} \right), \text{ for each } i \in [\T - 1]. 
    \end{equation}
\end{corollary}

\begin{proof}
    By \cref{prop:opt-beta1}, in order that $\Phi_s(P^*)$ has polynomial dependence on $d$, we must have step size $h \downarrow 0$. Hence, we can assume $h \leq 1$, which then implies that $\beta_1 = O(  (\log D) / D^2 )$. The condition on $\beta_{i+1} / \beta_i$ follows from~\cref{prop:opt-beta-ratio}. 
\end{proof}

\section{Discussion}\label{sec:disc}
One major limitation of the analysis conducted in this work is the suboptimal complexity bound obtained when $Q_i$ is the MALA proposal. As discussed in Remark~\ref{rmk:order}, this is due to the potentially suboptimal approach we take to bounding the $s$-conductance of $P_{i, j}$ and the squaring of   $\Gap(\barP)$ when applying the decomposition bound given in \cref{th:conduct-decomp}  with the mixing time bound in~\eqref{eq:s-conduct-tv-rate}. The latter can likely be overcome by refining the mixing time bound to incorporate the decomposition of $\Phi_s(P)$ or generalizing Cheeger's inequality to connect $s$-conductance with approximate or restricted spectral gaps~\citep{atchade2021approximate, garg2025restricted}. The former seems more challenging to improve. Existing sharp $s$-conductance bounds for MALA with log-concave targets are obtained by comparing the MALA dynamics with Langevin diffusion~\citep{wu2022minimax, chewi2021optimal}, which leads to a step size $h = \Theta(d^{-1/2})$, and it is unclear whether such techniques extend to our setting where the potential used in the MALA proposal is biased. In fact, whether MALA does provide any improvement over RWM in the context of simulated tempering with mixture targets remains an open question. 

We also note that all existing complexity results for simulated tempering rely on certain decomposition bounds for spectral gap or conductance in multiplicative form, e.g., the one given in~\eqref{eq:gap-decomp}. However, as shown in~\citet{pillai2017elementary}, such multiplicative bounds are likely suboptimal and, in certain scenarios, could be replaced by tighter, additive bounds. Hence, even our complexity bound for simulated tempering with RWM given in~\eqref{eq:comp-rwm} may be further improved to achieve a better dependence on $d$. 

Extending our analysis to other sampling algorithms is another interesting direction for future work. In recent years, non-reversible tempering algorithms have drawn growing interest~\citep{syed2022non}, but their complexity is not fully understood. While non-reversible samplers are often considered more efficient than reversible ones, the recent work of~\citet{roberts2025quantifying} suggests that their asymptotic complexities are likely of the same order. 
Other extensions include parallel tempering and sequential Monte Carlo algorithms, which exhibit dynamics similar to that of simulated tempering~
\citep{lee2023improved, lee2024convergence}, and our techniques may be adapted to analyze them.

\section*{Acknowledgments}
The author would like to thank Jhanvi Garg, Krishnakumar Balasubramanian and the participants of the conference ``The Fast and Curious 2: MCMC in Action'' at the University of Toronto (where the author presented a preliminary version of this work) for helpful discussions. 
This work was supported in part by NSF through grants DMS-2311307 and DMS-2245591.

\bibliographystyle{plainnat}
\bibliography{Reference}

\newpage 

\appendix

\begin{center}
    \LARGE\textbf{Appendices}
\end{center}

\section{Auxiliary Results}\label{sec:aux}

\subsection{Properties of s-conductance}
\begin{lemma}\label{lm:s-conduct}
Let $P$ be a transition kernel on $(\cX, \cB(\cX))$ reversible w.r.t. the distribution $\Pi$. 
\begin{enumerate}[label=(\roman*)]
    \item For any $A \in \cB(\cX)$ and $s \in [0, 1/2)$, 
    \begin{equation}
        V_ \Pi(\ind_A) - s \leq \frac{ \cE_P(\ind_A) }{ \Phi_s(P)}. 
    \end{equation} 
    \item  Suppose there exists some $A_0 \in \cB(\cX)$ such that $ \Pi(A_0) = 1/2$. Then, for any $s \in [0, 1/2)$, 
    \begin{equation}
      \frac{\Gap(P)}{2} \leq   \Phi_s(P) \leq \frac{1}{1-2s}. 
    \end{equation}
    \item $\Phi_s(P) \leq \Phi_{s'}(P)$ whenever $0 \leq s \leq s' < 1/2$. 
\end{enumerate}
\end{lemma}

\begin{proof}
Consider part (i), and fix an arbitrary $A \in \cB(\cX)$.
We always have $P(A, A^c) = P(A^c, A)$, which holds even if $P$ is not reversible. Hence, it follows from the definition of $\cE_P$ that 
\begin{equation}
    \cE_P(\ind_A) = \frac{1}{2} \left\{P(A, A^c) + P(A^c, A) \right\} = P(A, A^c) = P(A^c, A). 
\end{equation} 
If $ \Pi(A) \leq s$,  the asserted inequality in part (i) holds trivially. If $ \Pi(A) \in (s, 1/2]$, using  $V_ \Pi(\ind_A) =  \Pi(A)  \Pi(A^c)$ we get 
\begin{equation} 
    V_ \Pi(\ind_A) - s \leq     \Pi(A) - s \leq  \frac{  P(A, A^c) }{  \Phi_s(P)} = 
     \frac{  \cE_P(\ind_A) }{  \Phi_s(P)}.  
\end{equation} 
If $ \Pi(A) > 1/2$, we use $ V_ \Pi(\ind_A) - s \leq     \Pi(A^c) - s $ and then apply the same argument. 

For part (ii), the lower bound follows from part (i), since for any $A$ with $ \Pi(A) \in(s, 1/2]$, 
\begin{equation}
    \frac{ \cE_P(\ind_A)}{  \Pi(A) - s }  \geq     \frac{ \cE_P(\ind_A)}{   \Pi(A) }   \geq \frac{ \cE_P(\ind_A)}{ 2  V_ \Pi(A) } \geq \frac{\Gap(P)}{2}. 
\end{equation}
The upper bound follows by choosing $A = A_0$. 

To prove Part (iii),  we note that since $s' > s$,  
\begin{equation}
   \frac{P(A, A^c)}{\Pi(A) - s'} \geq \frac{P(A, A^c)}{\Pi(A) - s}. 
\end{equation}
Taking infimum on both sides over all $A$ such that $\Pi(A) \in (s', 1/2] \subset (s, 1/2]$, we get the claimed inequality. 
\end{proof}

\begin{lemma}\label{lm:s-conductance-isoperi}
    Let $\pi(x) \propto e^{-f(x)}$ be a probability density function on $\bbR^d$,  where $f$ is $m$-strongly convex.  Let $P$ be the transition kernel of a Markov chain on $\bbR^d$ reversible with respect to $\pi$. Suppose that for some convex $\cX_0 \subset \bbR^d$ and $\xi \in (0, 1/\sqrt{m}) $,  
    \begin{align}
       \Pi(\cX_0) \geq 1 - s, &\text{ for some } s \in [0, 1/2), \\ 
          \| P(x, \cdot) - P(y, \cdot) \|_{\TV}  \leq 3/4,  &\text{ whenever }
     x, y \in \cX_0 \text{ and } \| x - y\| \leq \xi. 
    \end{align}
   Then, $\Phi_s(P) \geq c \, \xi \sqrt{m}$ for some universal constant $c > 0$. 
\end{lemma}

\begin{proof}
This is essentially Lemma 6 of~\citet{dwivedi2019log}, which shows that  for any Borel set $A$, 
\begin{equation}
    P(A, A^c) \geq \frac{1}{16} \min \left\{ 1,  \, \frac{ \Pi(\cX_0)^2  \log 2}{8 } \xi \sqrt{m}  \right\} \min\left\{  \Pi(A \cap \cX_0), \,  \Pi(A^c \cap \cX_0) \right\}. 
\end{equation}
As shown in~\citet{dwivedi2019log}, this can be proved by first observing that  $\Pi$ restricted to $\cX_0$ is still $m$-strongly log-concave (and thus satisfies an isoperimetric inequality) and then applying a standard argument based on this inequality.  
Using the assumptions $\xi \leq 1/ \sqrt{m}$ and $ \Pi(\cX_0)  \geq  1 - s \geq 1/2$, we get 
\begin{equation}
   P(A, A^c)    \geq   \frac{  \log 2}{512} \, \xi \sqrt{m} \, \min\left\{  \Pi(A) - s, \,  \Pi(A^c ) - s \right\}. 
\end{equation} 
Since $A$ is arbitrary, this proves that $\Phi_s(P) \geq c\,  \xi \sqrt{m}$ with $c = (\log 2)/512$.  
\end{proof}

\subsection{Comparison Results for the Simulated Tempering Chain} \label{sec:aux-compare-st}

The proof of the following lemma is omitted and can be found in~\citet{ge2018simulated}.  

\begin{lemma}[Lemma 7.3 of~\citet{ge2018simulated}]\label{lm:ge-compare}
Let $\pi^*(i, x)$ be given by~\eqref{eq:def-pi-star-ix} and $\pi(i, x)$ be given by~\eqref{eq:def-approx-pi}. For any $(i, x) \in [\T] \times \bbR^d$, 
\begin{equation}\label{eq:compare-two-pi}
    \wmin \leq \frac{\pi^*(i, x)}{\pi(i, x)} \leq \frac{1}{\wmin}, \text{ where } \wmin = \min_{j \in [\K]} w_j.   
\end{equation}  
\end{lemma}  

Next, we prove a result for comparing the acceptance rates of the chains $\tilde{P}$ and $P$. This result is well known in the MCMC literature and has also been used in~\citet{garg2025restricted}. It essentially says that when targeting a mixture distribution, it is always more advantageous to compute the acceptance probability in the Metropolis--Hastings schemes using the marginal densities. 
Recall the acceptance probabilities for $P$ and $\tilde{P}$: 
\begin{align}
 a_{\tilde{P}}( (i, x), (i, x') )    =   1 \wedge \frac{ \pi(i, x') q_i(x', x)}{\pi(i, x) q_i(x, x')},   &\quad 
  a_{\tilde{P}}( (i, x), (i', x)  ) =   1 \wedge \frac{ \pi(i', x) q(i', i)}{\pi(i, x) q(i, i) },    \\
    a_{P}( (i, j, x), \, (i, j, x') )  =   1 \wedge \frac{ \pi(i, j, x') q_i(x', x)}{\pi(i, j, x) q_i(x, x')},  &\quad
    a_{P}( (i, j, x), \, (i', j, x) )  =   1 \wedge \frac{ \pi(i', j, x)  q(i', i)}{\pi(i, j, x) q(i, i')}.  
\end{align}

\begin{lemma}\label{lm:compare-acc}
Let $\pi(i, j, x)$ be given by~\eqref{eq:def-joint-pi}  and $a_P, a_{\tilde{P}}$ be as defined above. Then, 
\begin{align}
    \sum_{j \in [\K]}    \pi(i, j, x)   a_P( (i, j, x), \, (i', j, x) ) 
    &\leq \pi(i,  x) a_{\tilde{P}}( (i, x), \, (i', x) ),  \\
  \sum_{j \in [\K]}    \pi(i, j, x)  a_P( (i, j, x), \, (i, j, x') )  &\leq \pi(i, x)  a_{\tilde{P}}( (i,   x), \, (i,  x') ).
\end{align}
\end{lemma}

\begin{proof}
For any joint probability density function $p(y, z) \geq 0$ and constant $c \in \bbR$,  using Jensen's inequality  and  the concavity of minimum,  we find that 
 \begin{equation}
   \int p(z \mid y)  \min\left\{1,  \frac{c\, p(y', z)   }{p(y, z)  } \right\}  \dd z   \leq    \min\left\{1,  \int \frac{c\, p(y', z)   }{p(y, z)  }    p(z \mid y) \dd z \right\}=  \min\left\{ 1, \,  \frac{ c\,   p(y') }{  p(y) } \right\}. 
  \end{equation}
Since $p(z \mid y) = p(y, z) / p(y)$, we get 
\begin{equation}
      \int p(y, z) \min\left\{1,  \frac{c\, p(y', z)   }{p(y, z)  } \right\}  \dd z   \leq   p(y ) \min\left\{ 1, \,  \frac{ c\,   p(y') }{  p(y) } \right\}. 
\end{equation}

To prove the first asserted inequality, let $y = (i, x), y' = (i', x)$ and $z = j$, $p(y, z) = \pi(i, j, x)$ and $c = q(i', i) / q(i, i')$, which yields
\begin{equation}
\sum_{j \in [\K]} \pi(i, j, x) \min\left\{ 1, \,  \frac{\pi(i', j, x)}{\pi(i, j, x)}  \frac{q(i', i)}{q(i, i')}  \right\}
\leq  \pi(i, x)\min\left\{ 1, \,  \frac{\pi(i',  x)}{\pi(i,  x)}  \frac{q(i', i)}{q(i, i')}  \right\}. 
\end{equation}
The second claimed inequality can be proved similarly.  
\end{proof}

\subsection{Other Auxiliary Results} 
\label{sec:prelim}

\begin{lemma}\label{lm:KL}
Let $f \geq 0$ and $p_\beta$ denote a probability density function on $\bbR^d$ given by 
\begin{equation}
    p_\beta(x) = e^{ - \beta f(x) - \psi(\beta) }, \text{ where } \psi(\beta) = \log \int_{\bbR^d} e^{-\beta f(x)} \dd x. 
\end{equation}
Let $v(\beta)$ denote the variance of $f(X)$ with $X \sim p_\beta$. 
Then, 
\begin{equation}
    \KL( p_\beta \,\|\, p_{\beta'} ) = \int_\beta^{\beta'} (\beta' - z)   v(z) \, \dd z. 
\end{equation}
\end{lemma}

\begin{proof}
This representation of KL divergence is classical; see, e.g.,~\citet[Chap. 3.5]{amari2000methods}, for a more general result. 
To prove the claim,   first observe that $p_\beta$ can be seen as an exponential family with sufficient statistic $-f(X)$ and parameter $\beta$. Then, by a standard result for exponential families~\citep[Prop. 3.2]{shao1999mathematical},  if $X \sim p_\beta$, 
\begin{equation}\label{eq:fisher-info}
    \sE[ f(X) ] =  -\psi'(\beta),  \quad v(\beta) = \mathrm{Var}(f(X)) =   \psi''(\beta) \geq 0. 
\end{equation} 

By the definition of KL divergence, 
\begin{align*}
    \KL( p_\beta \,\|\, p_{\beta'} ) &= \int_{\bbR^d} p_\beta(x) \log \frac{p_\beta(x)}{p_{\beta'}(x)}  \dd x \\
    &= \int_{\bbR^d} e^{ - \beta f(x) - \psi(\beta) } \left[ (\beta' - \beta) f(x) + \psi(\beta') - \psi(\beta) \right]  \dd x  \\
    &= \psi(\beta') - \psi(\beta) + (\beta' - \beta)  \int_{\bbR^d} e^{ - \beta f(x) - \psi(\beta) }  f(x)  \dd x  \\
    &= \psi(\beta') - \psi(\beta) - (\beta' - \beta) \psi'(\beta), 
\end{align*}
where the last step follows from~\eqref{eq:fisher-info}. 
Expressing $\psi(\beta') - \psi(\beta)$ as an integral of $\psi'$, we get 
\begin{align*}
    \KL( p_\beta \,\|\, p_{\beta'} ) &=  \int_\beta^{\beta'} [ \psi'(y)  - \psi'(\beta) ] \dd y
    = \int_\beta^{\beta'} \left\{ \int_{\beta}^{y} \psi''(z) \dd  z \right\} \dd y
\end{align*}

Since 
$$\{ (y, z) \colon z \in [\beta, y], \, y \in [\beta, \beta'] \} = \{ (y, z) \colon z \in [\beta, \beta'], \, y \in [z, \beta'] \},$$
the claim then follows from~\eqref{eq:fisher-info} and a routine calculation using Fubini's theorem. 
\end{proof}

\newpage 

\section{State Decomposition for Reversible Markov Chains}  \label{sec:proof-conduct-decomp}

We  recall the following result proved in~\citet{jerrum2004elementary}. 

\begin{proposition}[Theorem 1 of~\citet{jerrum2004elementary}]\label{th:jerrum}
Consider the setting of~\cref{def:decomp}. For any measurable and square-integrable function $g$, define 
\begin{equation}\label{eq:def-Pkl}
    \cE_{P}^{k, \ell}(g)  = \int_{x \in \cX_k} \int_{y \in \cX_{\ell}} [g(y) - g(x)]^2  \Pi(\dd x) P(x, \dd y), 
\end{equation}
and $\gamma = \max_{k \in [n]} \sup_{x \in \cX_k} P(x, \cX \setminus \cX_k)$. 
Then, we have 
\begin{align}
     \cE_P(g)  
    &=   \frac{1}{2} \sum_{k \neq \ell} \cE_{P}^{k, \ell}(g) + \sum_{k \in [n]} \bar{\pi}(k) \cE_{P_k}(g), \label{eq:decomp-E} \\  
    V_ \Pi(g)    &\leq \frac{3}{2 \Gap(\barP)}\sum_{k \neq \ell} \cE_{P}^{k, \ell}(g)   + \left[\frac{3 \gamma}{\Gap(\barP)} + 1\right] \sum_{k \in [n]} \bar{\pi}(k)  V_{ \Pi_k}(g).   \label{eq:decomp-V}
\end{align}
\end{proposition}

\begin{proof}
This is essentially Theorem~1 of~\citet{jerrum2004elementary}. For completeness, we give their proof here.  
Fix an arbitrary square-integrable function $g$, and let $ \Pi_k$ and $\bar{\pi}$ be as defined in~\eqref{eq:def-muk-barmu}. 
Define the function $\bar{g} \colon [n] \rightarrow \bbR$ by $\bar{g}(k) =  \Pi_k(g)$. 
The Dirichlet form of $P$ can be expressed as 
\begin{align*}
    \cE_P(g) &
    = \frac{1}{2} \sum_{k, \ell \in [n]} \cE_{P}^{k, \ell}(g). 
\end{align*} 
By the definition of $P_k$, we have  $\cE_{P}^{k, k}(g) = 2 \bar{\pi}(k) \,  \cE_{P_k}(g)$. Hence, we obtain~\eqref{eq:decomp-E}.

By the law of of total variance and the definition of spectral gap, 
\begin{equation}\label{eq:decomp0}
    V_ \Pi(g) = V_{\bar{\Pi}}(\bar{g}) + \sum_{k \in [n]} \bar{\pi}(k) V_{\Pi_k}(g) \leq  \frac{\cE_{ \barP}(\bar{g})  }{ \Gap(\barP)} + \sum_{k \in [n]}\bar{\pi}(k) V_{ \Pi_k}(g).   
\end{equation}  
Recalling that $\bar{g}(k) =  \Pi_k(g)$ and using the triangle inequality, we find that 
\begin{equation}\label{eq:decomp2} 
   \cE_{ \barP}(\bar{g})   = \sum_{k \neq \ell}  [\barg(\ell) - \barg(k) ]^2\bar{\pi}(k) \barP(k, \ell)  \leq \frac{3}{2} \left( \bar{E}_1 + \bar{E}_2 + \bar{E}_3  \right), 
\end{equation} 
where 
\begin{align*}
    \bar{E}_1 &=    \sum_{k \neq \ell}    [  \Pi_k(g) -  \Pi_{k \mid \ell}(g) ]^2    \bar{\pi}(k) \barP(k, \ell),  \\
    \bar{E}_2 &=    \sum_{k \neq \ell}      [  \Pi_{k \mid \ell}(g) -  \Pi_{\ell \mid k}(g) ]^2  \bar{\pi}(k) \barP(k, \ell), \\
    \bar{E}_3 &=    \sum_{k \neq \ell}   [  \Pi_\ell(g) -  \Pi_{\ell \mid k}(g) ]^2   \bar{\pi}(k) \barP(k, \ell). 
\end{align*}

In the above definitions of $\bar{E}_1, \bar{E}_2, \bar{E}_3$, the notation $ \Pi_{k \mid \ell}$ denotes the conditional distribution of $Y_0$ given $\{Y_0 \in \cX_k, Y_1 \in \cX_\ell\}$, where  $Y_0 \sim  \Pi$ and $Y_1 \sim P(Y_0, \cdot)$.
That is, $ \Pi_{k \mid \ell}$  is a probability measure with support $\cX_k$ defined by 
\begin{equation}\label{eq:def-mu-k-ell}
     \Pi_{k \mid \ell} (A) = \frac{ \int_{A \cap \cX_k} P(x, \cX_\ell)   \Pi_k(\dd x) }{ \int_{\cX_k}  P(x, \cX_\ell)   \Pi_k(\dd x) }  = \frac{ 1}{  \barP(k, \ell)}\int_{A \cap \cX_k}  P(x, \cX_\ell)   \Pi_k (\dd x). 
\end{equation} 
By the reversibility of $P$, $ \Pi_{\ell \mid k}$ can also be viewed as the conditional distribution of $Y_1$ given $\{Y_0 \in \cX_k, Y_1 \in \cX_\ell\}$.  Hence, 
\begin{align*}
     \Pi_{k \mid \ell}(g) -  \Pi_{\ell \mid k}(g) &= \frac{1}{\barP(k, \ell)} \int_{x \in \cX_k, y \in \cX_\ell} \left[ g(x) - g(y) \right]  \Pi_k(\dd x) P(x, \dd y). 
\end{align*} 
It then follows from the elementary inequality  $\sE[Z]^2 \leq \sE[Z^2]$ that  
\begin{equation}\label{eq:decomp-E2}
    \bar{E}_2  \leq    \sum_{k \neq \ell}     \bar{\pi}(k)  \int_{x \in \cX_k, y \in \cX_\ell} \left[ g(x) - g(y) \right]^2  \Pi_k(\dd x) P(x, \dd y) = \sum_{k \neq \ell} \cE_{P}^{k, \ell}(g). 
\end{equation}

Since $\bar{P}$ is also reversible, $\bar{E}_1 = \bar{E}_3$. To bound $\bar{E}_1$, we use $\sE[Z]^2 \leq \sE[Z^2]$ to get 
\begin{align*}
  [ \Pi_{k \mid \ell}(g)-  \Pi_k(g) ]^2 &= 
  \left[  \Pi_{k \mid \ell}(g -  \Pi_k(g) ) \right]^2   \leq     \Pi_{k \mid \ell}\left( [g -  \Pi_k(g)]^2  \right). 
\end{align*}
Hence, using~\eqref{eq:def-mu-k-ell} and the definition of $\gamma$ 
\begin{align}
 \bar{E}_1 &\leq   \sum_{k \neq \ell}  \Pi_{k \mid \ell}\left( [g -  \Pi_k(g)]^2  \right)   \bar{\pi}(k) \barP(k, \ell) \\
 & = \sum_{k \in [n]}\bar{\pi}(k) \sum_{\ell \neq k} \int  [g(x) -  \Pi_k(g)]^2 P(x, \cX_\ell)  \Pi_k(\dd x) \\
 & \leq  \gamma \sum_{k \in [n]} \bar{\pi}(k)  V_{ \Pi_k}(g).  \label{eq:decomp-E1}
\end{align}

Since $\bar{E}_1 = \bar{E}_3$, we obtain~\eqref{eq:decomp-V} by combining~\eqref{eq:decomp0},~\eqref{eq:decomp2},~\eqref{eq:decomp-E2} and~\eqref{eq:decomp-E1}. 
\end{proof} 

Below is a more general version of the state decomposition bound given in \cref{th:conduct-decomp}. When $\gamma$ and $\Gap(\barP)$ have the same order, it yields a lower bound on the $s$-conductance of order $ \min\{ \Gap(\barP), \Psi(c s \Gap(\barP)  ) \} $ for some constant $c > 0$. 

\begin{theorem}\label{th:conduct-decomp-general}
In the setting of~\cref{def:decomp}, for $s \in [0, 1/2)$, we have 
\begin{equation}
    \Phi_s(P) \geq \min \left\{  \frac{\Gap(\barP)}{6}, \;   \frac{\Gap(\barP) } { 6 \gamma + 2 \Gap(\barP)}  \, \Psi\left(   \frac{\Gap(\barP) \, s } { 6 \gamma + 2 \Gap(\barP)}  \right)
    \right\},  
\end{equation}
where  
$\Psi(s) = \min_{k} \Phi_s(P_k)$  and  $\gamma = \max_{k \in [n]} \sup_{x \in \cX_k} P(x, \cX \setminus \cX_k).$
\end{theorem}
\begin{proof}
    The proof is almost identical to that of \cref{th:conduct-decomp}. The only change is that we do not assume or use the bounds $\Gap(\barP) \leq 1$ and $\gamma \leq 1$. 
\end{proof}

\newpage 

\section{Canonical Paths for the Projected Chain}

Consider the transition kernel $\barP$ of the projected chain,  
\begin{equation}
    \barP( (i, j), \, (i', j') ) 
    =  \int_{\bbR^d}   \pi_{i,j}(x) P( (i, j, x), \, (i', j', x ) ) \dd x, \text{ if } (i,j) \neq (i', j'), 
\end{equation}
where  $P$ denotes the auxiliary Metropolis--Hastings chain constructed in \cref{sec:second-aux-chain}.  

\begin{lemma}\label{lm:canonical-path}
For the projected chain $\barP$ on $[\T] \times [\K]$ with stationary distribution $\pi(i, j) = r_i w_j$,  \begin{equation}
    \Gap(\barP)  \geq \frac{1}{2 \T}  \min\{ \Lambda_1, \Lambda_2 \}, 
\end{equation}
where 
\begin{align}
    \Lambda_1 = \min_{i \in [\T - 1], j \in [\K]}  r_i \, \barP( (i, j), \, (i+1, j) ),  \quad 
    \Lambda_2 = \min_{j, j' \in [\K]} \frac{ \min_{i \in [\T]} r_i }{w_{j'} }  \barP( (1, j), \, (1, j') ).  
\end{align}
\end{lemma}

\begin{proof}
We apply the same canonical path argument as used in~\citet[Thm. 7.1]{ge2018simulated} and~\citet[Lem. 19]{garg2025restricted}. 
Denote the state space of $\barP$ by $\cZ =  [\T] \times [\K]$.
Consider a graph $\cG$ with node set $\cZ$ and directed edge set 
\begin{equation}
   \mathrm{Edge}(\cG) =  \left\{ ( (i, j), \, (i', j) ) \in \cZ^2 \colon   |i - i'| = 1 \right\} \cup \left\{ ( (1, j), \, (1, j') ) \in \cZ^2 \colon j \neq j'  \right\}. 
\end{equation} 
We define the canonical path from $(i, j)$ to $(i', j')$ as the shortest directed path on $\cG$: 
\begin{itemize}[label=$\circ$]
    \item If $j = j'$ and $i \leq i'$,  the canonical path from $(i, j)$ to $(i',j)$ is given by 
    $$ (i, j) \rightarrow (i+1, j) \rightarrow \cdots \rightarrow (i', j).$$ 
    \item If $j \neq j'$, the canonical path from $(i, j)$ to $(i',j')$ is given by 
    $$ (i, j) \rightarrow (i-1, j) \rightarrow \cdots \rightarrow (1, j) \rightarrow (1, j') \rightarrow (2, j') \rightarrow \cdots \rightarrow (i', j').$$
\end{itemize}

Clearly, any such canonical path has length (i.e., number of edges traversed) bounded by $2 \T$. Given an edge $e  \in  \mathrm{Edge}(\cG)$, let $\cA(e)$ denote the set of all pairs $(z, z') \in \cZ$ such that the path from $z$ to $z'$ traverses $e$, and define 
$\cL(e) = \sum_{ (z, z') \in  \cA(e)} \pi(z) \pi(z')$.  
By the canonical path method~\citep{sinclair1992improved}, 
\begin{equation}\label{eq:canon-path}
    \Gap(\barP)^{-1} \leq \max_{e = ( (i, j), \, (i', j') ) \in  \mathrm{Edge}(\cG)} \frac{2 \T \cL(e)}{ \pi(i, j) \barP( (i, j), \, (i', j') ) }. 
\end{equation}

If $e = ( (i, j), \, (i+1, j)) $, then the path from $z$ to $z'$ traverses $e$ only when $z' = (i', j)$ for some $i' \geq i + 1$. Hence, 
\begin{equation}
    \frac{\cL(e)}{\pi(i, j)} \leq  \frac{1}{\pi(i, j)} \sum_{i' \geq i + 1} \pi(i', j) \leq \frac{1}{\pi(i, j)} \sum_{i' \in [\T]} \pi(i', j) = \frac{w_j}{r_i w_j} = \frac{1}{r_i}. 
\end{equation}
  
If $e = ( (1, j), \, (1, j')) $,  then the path from $z$ to $z'$ traverses $e$ only when $z = (i, j), z' = (i', j')$  for some $i, i' \in [\T]$. Hence, 
\begin{equation}
     \frac{\cL(e)}{\pi(i, j)}   \leq \frac{1}{\pi(i, j)} 
 \sum_{i, i'} \pi(i, j) \pi(i', j') = \frac{1}{\pi(i, j)} \left( \sum_{i} \pi(i, j) \right) \left( \sum_{i'} \pi(i, j') \right) = \frac{ w_j w_{j'} }{r_i w_j} = \frac{w_{j'}}{ r_i}.
\end{equation}
The claimed bound then follows from~\eqref{eq:canon-path}. 
\end{proof} 

\newpage 

\section{Conductance Bounds for the Restricted Chains}\label{sec:supp-Pij}

In this section, we prove a sequence of auxiliary results used to establish the $s$-conductance bounds for the restricted chains $P_{i, j}$ considered in \cref{sec:conduct-restrict} when the MALA proposal is used. 

\subsection{Auxiliary Results for Smooth and Strongly Convex Functions} \label{sec:prelim-convex}

We first recall some useful auxiliary results for smooth and strongly convex functions and strongly log-concave distributions. 

\begin{definition}\label{def:smooth-convex}
Let $f \colon \bbR^d \rightarrow \bbR$ be differentiable. We say $f$ is $L$-smooth for $L>0$ if 
\begin{equation}
    f(y) - f(x) - \nabla f(x)^\top (y - x) \leq \frac{L}{2} \| y - x \|^2, \quad \forall x, y \in \bbR^d.  
\end{equation} 
Similarly, we say $f$ is $m$-strongly convex for $m>0$  if 
\begin{equation}
    f(y) - f(x) - \nabla f(x)^\top (y - x) \geq \frac{m}{2} \| y - x \|^2, \quad \forall x, y \in \bbR^d.  
\end{equation} 
We say a probability distribution is $m$-strongly log-concave, if its density is proportional to $e^{-f(x)}$ for some $m$-strongly convex $f$. 
\end{definition}

\begin{lemma}[Theorem 2.1.5 of~\citet{nesterov2013introductory}]\label{lm:smooth}
Let $f \colon \bbR^d \rightarrow \bbR$ be differentiable and $L$-smooth. For any $t \in [0, 1],  x, y \in \bbR^d$, 
\begin{align}
(\nabla f (x) - \nabla f(y))^\top (x - y) &\leq L \| x - y\|^2, \\ 
t f(x) + (1 - t) f(y) &\leq f ( tx + (1-t) y ) + \frac{t (1-t) L}{2} \| y - x\|^2. 
\end{align}
If $f$ is also convex, then for any $ x, y \in \bbR^d$, 
\begin{align}
     \| \nabla f (x) - \nabla f(y)\| &\leq L \| x - y\|, \\ 
     \frac{1}{L}  \| \nabla f (x) - \nabla f(y)\|^2 &\leq  (\nabla f (x) - \nabla f(y))^\top (x - y). 
\end{align}
\end{lemma} 

\begin{lemma}\label{lm:mala-mean-bound}
   Let $f \colon \bbR^d \rightarrow \bbR$ be differentiable, $L$-smooth and convex. For any $h \in (0, 1/L]$,  
   \begin{equation}
       \| (x - h \nabla f(x) ) - (y - h \nabla f(y) ) \| \leq \| x - y\|. 
   \end{equation}
\end{lemma}
\begin{proof}
    This is established in Section 5.5.1 of~\citet{dwivedi2019log}, as part of the  proof of Lemma 7 therein. 
\end{proof}

\begin{lemma}\label{lm:log-concave}
Let $\Pi$ be an $m$-strongly log-concave distribution on $\bbR^d$. For any integrable $g$, 
\begin{align}
 & \text{Poincar\'{e} inequality: } \quad \quad   V_\Pi(g) \leq \frac{1}{m} \Pi \left(  \| \nabla g\|^2  \right), \\
  & \text{log-Sobolev inequality: } \quad    \mathrm{Ent}_\Pi(g^2) \leq \frac{2}{m}  \Pi \left(  \| \nabla g\|^2  \right),  
\end{align}
where $ \mathrm{Ent}_\Pi(\tilde{g}) =  \Pi \left( \tilde{g} \log \tilde{g} \right) - \Pi(\tilde{g}) \log \Pi(\tilde{g})$ denotes the entropy. 
Further,  letting $x^*$ denote the unique mode of $\Pi$, we have 
\begin{align}
    \int |x - x^*|^2 \Pi(\dd x) &\leq \frac{d}{m},  \label{eq:durmus} \\ 
    \Pi\left( \cB_s   \right) &\leq s, \label{eq:concentrate-m-convex} 
\end{align} 
where $s \in (0, 1)$ and 
   \begin{equation}
     \cB_s = \left\{x \in \bbR^d \colon \|x - x^*\| \geq  2 \sqrt{ \frac{  d \vee (2 \log s^{-1} ) }{m} } \right\}. 
    \end{equation} 
\end{lemma}
\begin{proof}
    For the Poincar\'{e} and log-Sobolev inequalities, see Proposition 10.1 of~\citet{saumard2014log}. 
    The inequality~\eqref{eq:durmus} is proved in Proposition 1 of~\citet{durmus2019high}. 

To prove~\eqref{eq:concentrate-m-convex}, we use that $\Pi$ satisfies log-Sobolev inequality with constant $m^{-1}$. Since any distribution satisfying a log-Sobolev inequality has Gaussian concentration of  Lipschitz functionals~\citep[Cor. 3.2]{BobkovGoetze1999} and  $x \mapsto \|x - x^*\|$ is $1$-Lipschitz, 
we have 
\begin{equation}
  \sP \left( \| X  - x^*\| - \sE \| X - x^* \| \geq t \right) \leq e^{ - m t^2  /2 }, \quad \text{ where } X \sim \Pi. 
\end{equation}
By~\eqref{eq:durmus}, $\sE \| X - x^* \| \leq \sqrt{d / m}$. Hence, letting $t^2 = 2 m^{-1}\log s^{-1} $, 
we find that 
\begin{equation}
    \Pi\left( \{x \colon \|x -x^* \| \geq r_s \} \right) \leq s,  \text{ where } r_s = \sqrt{\frac{d}{m}} + \sqrt{\frac{2 \log s^{-1}}{m}}. 
\end{equation} 
The claimed bound then follows.  
\end{proof}

\subsection{Gradient Bounds for the Local and Mixture Potential}

Consider the function $U \colon \bbR^d \rightarrow \bbR$ defined in~\eqref{eq:target-pi-star}, 
\begin{equation}
    U(x) = -\log \sum_{j=1}^{\K} w_j e^{ - f(x - \mu_j)}, 
\end{equation}
where $\mu_j \in \bbR^d, w_j \geq 0$ and $\sum_{j = 1}^{\K} w_j = 1$.
Let $D = \max_{j  \in [\K]} \| \mu_j \|.$  
The function $f$ characterizes the potential of each component distribution and thus we call it the ``local potential'', while $U$ is the potential of the mixture density. 

\begin{lemma}\label{lm:U-bound}
Let $f \colon \bbR^d \rightarrow \bbR$ be differentiable, $L$-smooth and convex.  Then, 
\begin{align}
\| \nabla U(x) - \nabla f(x ) \| &\leq L D,  \label{eq:lm-grad-U} \\ 
    U(y) - U(x) - \nabla U(x)^\top (y - x) &\leq \frac{L}{2} \| y - x \|^2, \quad \forall x, y \in  \bbR^d.   \label{eq:lm-smooth-U}
\end{align}
\end{lemma}

\begin{proof}
A direct calculation yields that 
\begin{equation}
    \nabla U(x) =      \sum_{j=1}^{\K}  
    \omega(x, j) \nabla f(x - \mu_j), \text{ where } \omega(x, j) = \frac{ w_j e^{ -f (x - \mu_j)} }{  e^{- U(x) } }. 
\end{equation} 
Using $\sum_{j=1}^{\K} \omega(x, j) = 1$, we find that 
\begin{align*}
   \|   \nabla U(x) - \nabla f(x ) \|  
    \leq \sum_{j=1}^{\K} \omega(x, j)  \|  \nabla f(x - \mu_{j}) - \nabla f(x) \|. 
\end{align*} 
Then~\eqref{eq:lm-grad-U} follows from $\|  \nabla f(x - \mu_{j}) - \nabla f(x) \| \leq L \|  \mu_j \| \leq L D$. 

To prove that $U$ is $L$-smooth, we first use the $L$-smoothness of $f$ to get that for any $x, y \in \bbR^d$, 
\begin{equation}
    - f(y - \mu_j) \geq  - f(x - \mu_j) - \nabla f(x - \mu_j)^\top (y - x) - \frac{L}{2} \| y - x \|^2.  
\end{equation}
Then, by the definition of $U$ and Jensen's inequality, 
\begin{align}
  e^{-U(y)} 
  &\geq  \sum_{j=1}^{\K} w_j e^{  -f(x - \mu_j) - \nabla f(x - \mu_j)^\top (y - x) - \frac{L}{2} \| y - x \|^2}  \\
  &=  e^{- U(x) - \frac{L}{2} \| y - x \|^2} \sum_{j=1}^{\K} \omega(x, j) e^{   - \nabla f(x - \mu_j)^\top (y - x)  } \\
  & \geq e^{- U(x) - \frac{L}{2} \| y - x \|^2}  e^{   -  \sum_{j=1}^{\K} \omega(x, j) \nabla f(x - \mu_j)^\top (y - x)  } \\
  & = e^{- U(x) - \frac{L}{2} \| y - x \|^2}  e^{   - \nabla U(x)^\top (y - x)  }.
\end{align}
Taking logarithm on both sides yields~\eqref{eq:lm-smooth-U}. 
\end{proof}

\begin{lemma}\label{lm:all-gradients}
    Let $f$ satisfy \cref{ass:f}.  Define 
    \begin{equation}
         y(x, z) = x - h \nabla U(x) + \sqrt{2h / \beta } \, z. 
    \end{equation}
    Then, the following bounds hold for any $x, z \in \bbR^d, h > 0, \beta \in (0, 1], j \in [\K]$:
    \begin{align*}
       \sqrt{\beta}\, \| \nabla f(x - \mu_j ) \| &\leq L ( \sqrt{\beta} \|x\| + D), \\ 
      \sqrt{\beta}\, \| \nabla U(x) \| &\leq L ( \sqrt{\beta} \|x\| + D), \\ 
        \sqrt{\beta} \| \nabla U   (y(x, z)) \| &\leq \sqrt{2h} \, L \| z \|  +  (1 + hL) L ( \sqrt{\beta} \| x \| + D). 
    \end{align*}
\end{lemma}

\begin{proof}
By~\cref{lm:smooth},  $f$ has Lipschitz gradient. Since $f(0) = \nabla f (0) = 0$,  we have 
\begin{equation}
    \| \nabla f(x ) \| \leq  \| \nabla f(0) \|   + L \| x - 0 \| = L \| x \|, 
\end{equation}
which implies that 
\begin{equation}
   \sqrt{\beta} \| \nabla f(x - \mu_j ) \|    \leq \sqrt{\beta} L ( \|x\| + \|\mu_j\|) \leq L( \sqrt{\beta} \|x\| + \sqrt{\beta} D). 
\end{equation} 
Since $\beta \leq 1$, this proves the first claimed inequality. 
For the second, applying \cref{lm:U-bound} we get 
\begin{equation} 
   \sqrt{\beta} \| \nabla U(x) \| \leq 
   \sqrt{\beta} \left( \| \nabla f(x) \| + LD \right) \leq L (\sqrt{\beta} \| x \| + D). 
\end{equation}  

Consider  $\| \nabla U   (y(x, z)) \|$. By \cref{lm:U-bound} and triangle inequality,  
\begin{align}
\sqrt{\beta} \| \nabla U   (y(x, z)) \| &\leq  LD + \| \nabla f(y(x, z)) \| \\ 
   &\leq  \sqrt{\beta}  L D+  \sqrt{\beta} L  \, \| y(x, z) \|  \\ 
    &\leq  \sqrt{2h} \, L \|z\| +  L ( \sqrt{\beta} \| x \| + D) + h L  \sqrt{\beta} \| \nabla U(x)\|  \\   
   & \leq    \sqrt{2h} \, L \| z \|  +  L ( \sqrt{\beta} \| x \| + D) + h L^2 (\sqrt{\beta} \|x \| + D)  \\
   & \leq   \sqrt{2h} \, L \| z \|  +  (1 + hL) L ( \sqrt{\beta} \| x \| + D). 
\end{align}
\end{proof}

\begin{lemma}\label{lm:acc-mala}
    Let $f$ satisfy \cref{ass:f}.  Define 
    \begin{align}
         y(x, z) &= x - h \nabla U(x) + \sqrt{2h / \beta } \, z, \\ 
         q_\beta(x, y) &= \frac{1}{ (4 \pi h / \beta)^{d / 2} } \exp\left\{ - \frac{\beta}{4 h} \|y - x + h \nabla U(x)\|^2 \right\}.
    \end{align} 
    Suppose 
    \begin{equation}\label{eq:cond-h}
         0 <  h  \leq \frac{ c_h }{ L^2 ( \sqrt{\beta} \, \|x\| + D )^2 }, 
    \end{equation}
    for some $c_h > 0$.  
Then, 
\begin{align}
& \beta \left[ f(x - \mu_j) -  f(y(x, z) - \mu_j) \right]    \geq  - hL \|z\|^2  - (1 + hL) \sqrt{2 c_h} \| z \| -  \left( 1 + \frac{hL}{2}\right)  \sqrt{c_h}   , \\
& \log \frac{ q_\beta(y(x, z), x)}{q_\beta(x, y(x, z))}   
  \geq    - hL \left( 1 + \frac{hL}{2} \right)\|z\|^2 - (2 + 3 hL + h^2 L^2) \sqrt{\frac{c_h}{2}}  \| z \| -  \left( 1 + \frac{hL}{2} \right)^2    c_h. 
\end{align}  
\end{lemma}

\begin{proof} 
We begin by deriving bounds on $\| \nabla U \|$ and $\| \nabla f \|$ that will be used frequently.  Applying \cref{lm:all-gradients} with condition~\eqref{eq:cond-h},  we get 
\begin{align}
   \sqrt{\beta  h} \, \| \nabla f(x - \mu_j) \| &\leq  \sqrt{h} \,  L  ( \sqrt{\beta} \|x\| + D ) \leq \sqrt{c_h}, \label{eq:gradient-fx}\\ 
   \sqrt{\beta  h} \, \| \nabla U (x) \| & \leq \sqrt{h} \,  L  ( \sqrt{\beta} \|x\| + D ) \leq   \sqrt{c_h},  \label{eq:gradient-Ux} \\ 
   \sqrt{\beta  h} \, \| \nabla U   (y(x, z)) \| &\leq \sqrt{2} \, hL \|z\| + (1 + hL) \sqrt{c_h  }.  \label{eq:gradient-Uy}  
\end{align}

Since $f$ is $L$-smooth, by~\cref{lm:smooth}, we have  
\begin{align}
 & \quad \beta \left[ f(x - \mu_j) -  f(y(x, z) - \mu_j) \right]  \\
& \geq - \beta  \nabla f(x - \mu_j)^\top (y(x, z) - x) - \frac{ \beta  L}{2} \| y(x, z) - x \|^2  \\
& =  - \beta  \nabla f(x - \mu_j)^\top ( -h \nabla U(x) + \sqrt{2h/\beta} z) - \frac{ \beta  L}{2} \| -h \nabla U(x) + \sqrt{2h/\beta} z \|^2  \\ 
& =   \beta h \nabla f(x - \mu_j)^\top   \nabla U(x)    - \frac{\beta h^2 L}{2} \| \nabla U(x)\|^2 - h L \|z\|^2 +  \sqrt{ 2 \beta h} \left[ h L \nabla U(x) -  \nabla f(x - \mu_j) \right]^\top z.  
\end{align} 
Using the Cauchy--Schwarz inequality, ~\eqref{eq:gradient-fx} and~\eqref{eq:gradient-Ux}, we get
\begin{align}
  \beta h \nabla f(x - \mu_j)^\top   \nabla U(x)  &\geq - \beta h \| \nabla f(x - \mu_j) \| \,   \|  \nabla U(x) \|  \geq - c_h, \\
  - \frac{\beta h^2 L}{2} \| \nabla U(x)\|^2 &\geq - \frac{hL}{2} c_h, \\ 
  \sqrt{ 2 \beta h} \left[ h L \nabla U(x) -  \nabla f(x - \mu_j) \right]^\top z &
  \geq - (1 + hL) \sqrt{ 2 c_h }\|z\|. 
\end{align}
The first asserted inequality then follows. 

Consider the second inequality. A routine calculation yields 
\begin{align}
   \log \frac{ q_\beta(y(x, z), x)}{q_\beta(x, y(x, z))}   &=     - \frac{\beta } {4h} \left\|h  \nabla U  (y(x, z))  +  h \nabla U (x)  - \sqrt{2h / \beta  } \, z  \right\|^2  + \frac{\beta}{4h} \|  \sqrt{2h / \beta  } \, z \|^2     \\
& =   - \frac{\beta    h}{4} \left\|   \nabla U  (y(x, z))  +   \nabla U (x)  \right\|^2   + \sqrt{ \frac{\beta h }{2}}  \left[ \nabla U  (y(x, z))  +    \nabla U (x) \right]^\top z   \\
& \geq - \frac{\beta    h}{4} \left\|   \nabla U  (y(x, z))  +   \nabla U (x)  \right\|^2   - \sqrt{ \frac{\beta h }{2}}  \| \nabla U  (y(x, z))  +  \nabla U (x)  \| \,  \| z \|. 
\label{eq:mala-prop-1}
\end{align} 
By~\eqref{eq:gradient-Ux} and~\eqref{eq:gradient-Uy},
\begin{equation}
    \sqrt{\beta h} \left\|   \nabla U  (y(x, z))  +    \nabla U (x)  \right\|
    \leq \sqrt{2} \, h L \| z \| + (2 + hL) \sqrt{c_h }. 
\end{equation}
Hence, 
\begin{align}
   \log \frac{ q_\beta(y(x, z), x)}{q_\beta(x, y(x, z))}   
& \geq  -\frac{1}{4} \left\{ \sqrt{2} \, h L \| z \| + (2 + hL) \sqrt{c_h }  \right\}^2 -  \left\{  h L \| z \| + (2 + hL) \sqrt{ \frac{c_h}{2} }  \right\} \| z \|  \\ 
& \geq  - \frac{h^2 L^2}{2} \|z\|^2 -  \left( 1 + \frac{hL}{2} \right)^2    c_h - h L (2 + hL) \sqrt{\frac{ c_h}{2 } } \|z \|  \\
& \quad - hL \|z\|^2  - (2 + hL) \sqrt{\frac{c_h}{2}}  \| z \|.  
\end{align}
The second claimed inequality then follows after   simplification. 
\end{proof}

\end{document}